\documentclass[a4paper]{article}

\usepackage{dsfont,amsmath,amsthm,amssymb}
\usepackage{subfigure}
\usepackage{tikz}
\usetikzlibrary{shapes.symbols,arrows}

\usepackage{fullpage}

\newtheorem{theorem}{Theorem}
\newtheorem{proposition}{Proposition}
\newtheorem{lemma}{Lemma}

\DeclareMathOperator{\Rem}{Rem}

\DeclareMathOperator{\Step}{Step}
\DeclareMathOperator{\Red}{Red}

\DeclareMathOperator{\Start}{Start}
\DeclareMathOperator{\End}{End}
\DeclareMathOperator{\Edges}{Edges}

\newcommand{\EP}{{ep}}

\newcommand{\nc}{{\mathrm{ nc}}}
\newcommand{\cc}{{\mathrm {c}}}

\newcommand{\IR}{\mathds{R}}
\newcommand{\IN}{\mathds{N}}

\newcommand{\Bcnc}{B_{\mathrm {c}}}		%
\newcommand{\Bunu}{\tilde{B}}    	%
\newcommand{\IRmax}{\overline{\IR}}

\newcommand{\ito}{{i\!\to}}
\newcommand{\Pa}{{\mathcal{W}}}

\newcommand{\realrem}[2]{\mathbf{N}_{\geqslant {#1}}^{({#2})}}

\renewcommand{\le}{\leqslant}
\renewcommand{\leq}{\leqslant}
\renewcommand{\ge}{\geqslant}
\renewcommand{\geq}{\geqslant}

\begin{document}

\title{New Transience Bounds for Long Walks}

\author{Bernadette Charron-Bost\textsuperscript{1} \and
Matthias F\"ugger\textsuperscript{2} \and
Thomas Nowak\textsuperscript{3}
}

\date{
\textsuperscript{1} CNRS, LIX, \'Ecole polytechnique, 91128 Palaiseau, France\\
\textsuperscript{2} ECS Group, TU Wien, 1040 Wien, Austria\\
\textsuperscript{3} LIX, \'Ecole polytechnique, 91128 Palaiseau, France
}

\maketitle

\abstract{%
Linear max-plus systems describe the behavior of a large variety of
     complex systems.
It is known that these systems show a periodic behavior after an
     initial transient phase.
Assessment of the length of this transient phase provides important
     information on complexity measures of such systems, and so is
     crucial in system design.
We identify relevant parameters in a graph representation of these
     systems and propose a modular strategy to derive new upper bounds
     on the length of the transient phase.
By that we are the first to give asymptotically tight and potentially
     subquadratic transience bounds.
We use our bounds to derive new complexity results, in particular in
     distributed computing.
}

\section{Introduction}

The behavior of many complex systems can be described by a sequence of
     $N$-dimensional vectors~$x(n)$ that satisfy a recurrence relation of the
     form 
	\begin{align}
		\forall n\geqslant 1\ \ \forall i\in \{1,\dots,N\}\ \,: \ \ \
x_i(n) = \max_{j\in {\bf N}_i}\big(x_j(n-1)+A_{i,j}\big) 
		\label{eq:sys_a}
	\end{align}
where the $A_{i,j}$ are real numbers, and the ${\bf N}_i$ are subsets of 
     $\{1,\dots,N\}$.
For instance,~$x_i(n)$ may represent the time of the $n$th occurence of a certain
event~$i$ and the~$A_{i,j}$ the required time lag between the $(n-1)$th
occurence of~$j$ and the $n$th occurence of~$i$.
Notable examples are transportation and automated manufacturing
     systems~\cite{GBO98,CDQV85,DS90}, network
     synchronizers~\cite{malka:rajs,even:rajs}, and cyclic
     scheduling~\cite{HM95}.
Recently, Charron-Bost et al.~\cite{fr:sirocco,pr:sirocco} have shown
     that it also  encompasses the behavior of an important class of
     distributed algorithms, namely {\em link
     reversal algorithms\/}~\cite{GB87}, which can be used to
     solve a variety of problems~\cite{walter:welch} like routing~\cite{GB87},
	 scheduling~\cite{BG89}, distributed queuing~\cite{TH06,AGM10},
     or resource allocation~\cite{CM}.

Interestingly, recurrences of the form~\eqref{eq:sys_a} are linear in 
	the {\em max-plus algebra} (e.g., \cite{workplus}).
The fundamental theorem in max-plus linear algebra---an analog of the Perron-Frobenius 
	theorem---states that the sequence of powers of an irreducible max-plus matrix becomes 
	periodic after a finite index called the {\em transient\/} of the matrix.    
As an immediate corollary, any linear max-plus system with irreducible system matrix 
	is periodic from some index, called the {\em transient\/} of the system, which clearly 
	depends on the system's initial vector and  is at most equal to the transient 
	of the matrix of the system.  
For all the above mentioned applications, the study of the transient
	 plays a key role in characterizing the system performances:
For example, in the case of link reversal routing, the system transient is equal to 
	the time complexity of the routing algorithm.
Besides that, understanding matrix and system transients is of
	interest on its own for the theory of max-plus algebra.
	
Hartmann and Arguelles~\cite{hartmann:arguelles} have shown that
	the transients of matrices and linear systems are computable in polynomial 
	time.
However, their algorithms provide no analysis of the transient phase, 
	and do not hint at the parameters that influence matrix and system transients.
Conversely, upper bounds involving these parameters help to
	predict the duration of the transient phase, and to define
	strategies to reduce transients during system design.
From both numerical and methodological viewpoints, it is therefore important to
	determine accurate transience bounds. 
	
In this paper, we present two upper bounds on the transients of linear max-plus systems.
Our approach is graph-theoretic in nature: The problem of bounding from above
	the transient can be reduced to the study of walks in a specific graph.
More precisely, for every max-plus matrix~$A$, one considers the weighted directed
	graph~$G$ whose adjacency matrix is~$A$, and its {\em critical
    subgraph\/} which consists of the {\em critical cycles\/}, namely those
    cycles with maximal average weight.
The entries of the max-plus matrix power~$A^{\otimes n}$ are equal to
	the maximum weights of walks in~$G$ of length~$n$ between two fixed nodes, and 
	when redefining the weights of walks in a way that respects initial vector~$v$, 
	the entries of~$A^{\otimes n} \otimes v$ are maximum weights of walks
	 of length~$n$ starting from a fixed node.
The periodicity of matrix powers and linear systems stems from the fact that eventually the weights of
	critical cycles dominate the maximum weight walks.
	
We present a general graph-based strategy  whose core idea is a walk reduction~$\Red_{d,k}$, 
	which removes cycles from a walk while assuring that its length remains in the 
	same residue class modulo~$d$, and that node~$k$ rests on the walk.
The key property of~$\Red_{d,k}$ is an upper bound on the length of
	the reduced walk that is linear both in~$d$ and the number of
	nodes in the graph.	
The following step in our strategy consists in completing reduced walks with
	critical cycles of appropriate lengths.
For that, we propose two methods, namely the {\em repetitive\/} method and the 
	{\em explorative\/} method.
In the first one, the visit of the critical subgraph is confined to repeatedly
     follow only one closed walk whereas the second one consists in exploring
	one whole strongly connected component of the  critical subgraph.
That leads us to give two upper bounds on the transients of linear systems,
 	namely the {\em repetitive bound} and the {\em explorative bound}, which
	are incomparable in general.
We show that in the case of integer matrices, for a given initial
	vector, both our transience bounds for a $A$-linear system are both in $O(\lVert A \rVert \cdot N^3)$, 
	where $\lVert A\rVert$ denotes the difference of the maximum and minimum finite entries of~$A$.
We also show  that this is asymptotically tight. 
		
Another contribution of this paper concerns the relationship between
	matrix and system transients:  We prove that the transient of an
	$N\times N$ matrix~$A$ coincides with the  transient of an~$A$-linear system 
	with an initial vector whose norm is at most quadratic in~$N$, provided the latter 
	transient is sufficiently large.
In addition to shedding new light on transients, this result provides
	two upper bounds on matrix transients.
		
The problem of bounding the transients has already been studied
	(e.g.,~\cite{hartmann:arguelles,BG00,Sot03}), and the 
	best previously known bound has been given by Hartmann and 
	Arguelles~\cite{hartmann:arguelles}.
Their bound on system transients is, in general, incomparable with our repetitive and explorative bounds.
The significant benefit of our two new bounds is that each of them turns out to be linear 
	in the size of the system in various classes of linear max-plus systems whereas 
	Hartmann and  Arguelles' bound is intrinsically at least quadratic.
This is mainly due to the introduction of new graph parameters that
	 enables a fine-grained  analysis of the transient phase.
In particular, we introduce the notion of the {\em exploration penalty\/} of a graph~$G$ 
	as the least integer~$k$ with the property that, for every  $n\geqslant k$ divisible by the
	cyclicity of~$G$ and every node~$i$ of~$G$, there is a closed
	path  starting and ending at~$i$ of length~$n$.
One key point is then an at most quadratic upper bound on the
	exploration penalty which we derive from  the number-theoretic
	Brauer's Theorem~\cite{Bra42}.

Finally, we demonstrate how our general transience bound enables the performance analysis
	of a large variety of distributed systems.
First, we apply our results to the class of	{\em earliest schedules} 
	in cyclic scheduling: we show that for a large family of sets of tasks, 
	earliest schedules correspond to linear max-plus systems with irreducible
	matrices.
Thus we prove the eventual periodicity of such earliest schedules, and give two upper bounds
	on their transient phases.
Then we derive two transience bounds for a large
 	class of synchronizers, and we quantify how both our synchronizer bounds are better 
	than that given by Even and Rajsbaum~\cite{even:rajs} 
 	in their specific case of integer delays.
In the process, we show that our transience bounds are asymptotically tight.
Our results also apply to the analysis of the performance of distributed routers and schedulers
	based on the link-reversal algorithms:
We obtain  $O(N^3)$ transience bounds,
	improving the $O(N^4)$ bound established by Malka and
	Rajsbaum~\cite{malka:rajs}, and  $O(N)$ bounds for 
	such routers and schedulers when running in trees.
For link-reversal routers, eventual periodicity actually corresponds to termination,
	and an $O(N^2)$ bound on time complexity~\cite{BT05} directly follows from our transience bounds.

The paper is organized as follows.
Section~\ref{sec:prelim} introduces basic notions of graph theory and max-plus algebra.
In Section~\ref{sec:bounds:outline}, we elaborate a graph-based strategy to prove 
	transience bounds.
We show an upper bound on lengths of maximum weight walks that do not visit the critical 
	subgraph in Section~\ref{sec:criticalbound}.
Section~\ref{sec:red} presents a walk reduction that constitutes the core of our 
	strategy.
In Section~\ref{sec:explorationpenalty}, we introduce the notion of {\em exploration penalty\/} 
	and improve a theorem by Denardo~\cite{denardo} on the existence of arbitrarily long walks 
	in strongly connected graphs.
We derive two transience bounds, namely the explorative and the repetitive bound, in 
	Section~\ref{sec:bounds}.
We show how to convert upper bounds on the transients of max-plus systems to upper bounds 
	on the transients of max-plus matrices in Section~\ref{sec:matrix}.
We discuss our results, by comparing them to previous work and by applying them to the analysis of 
	various complex systems, in Section~\ref{sec:discussion}.

\section{Preliminaries}\label{sec:prelim}  
This section introduces  definitions and classical results needed in the rest of the
paper.
We denote by~$\IN$ the set of nonnegative integers and by~$\IN^*$ the set of
positive integers.

\subsection{Graphs}

A {\em directed graph\/}~$G$ is a pair~$(V,E)$ where~$V$ is
     a nonempty finite set and $E\subseteq V\times V$.
The elements of~$V$ are the {\em nodes\/} of~$G$ and the elements
     of~$E$ the {\em edges\/} of~$G$.
In this paper, we refer to directed graphs simply as {\em graphs}.

A {\em walk\/}~$W$ in $G$ is a triple $W=(\Start,\Edges,\End)$
     where $\Start$ and $\End$ are nodes in $G$, $\Edges$ is a
     sequence $(e_1,e_2,\dots,e_n)$ of edges  $e_l=(i_l,j_l)$ such
     that $j_l=i_{l+1}$ if $1\leqslant l\leqslant n-1$,
     $i_1=\Start$ and $j_n=\End$ if the sequence $\Edges$ is nonempty, and $\Start=\End$ if
     the sequence $\Edges$ is empty. 
We define the operators $\Start$, $\Edges$, and $\End$ on the set of
     walks by setting $\Start(W)=\Start$, $\Edges(W)=\Edges$, and
     $\End(W)=\End$.
We call $\Start(W)$ the {\em start node\/} of~$W$ and $\End(W)$ the {\em end node\/}
of~$W$.
The {\em length\/}~$\ell(W)$ of~$W$ is defined as the length of the sequence
$\Edges(W)$. 
Walk~$W$ is {\em closed\/} if $\Start(W)=\End(W)$.
Walk~$W$ is {\em empty\/} if the sequence $\Edges(W)$ is empty.
A walk~$W$ is empty if and only
     if $\ell(W)=0$.

For two walks~$W$ and~$W'$, we say that~$W'$ is a {\em
     prefix\/} of~$W$ if $\Start(W)=\Start(W')$ and
     the sequence $\Edges(W')$ is a prefix of $\Edges(W)$.
We say that~$W'$ is a {\em postfix\/} of~$W$ if
     $\End(W)=\End(W')$ and the sequence $\Edges(W')$ is a postfix of
     $\Edges(W)$.
We call~$W'$ a {\em subwalk\/} of~$W$ if it is the postfix of some
     prefix of~$W$.
A subwalk~$W'$ of~$W$ is a {\em proper\/} subwalk of~$W$ if $W'\neq W$.
We say a node~$i$ is a {\em node of walk $W$\/} if there exists a
     prefix~$W'$ of~$W$ with $\End(W')=i$.
For two walks $W_1$ and $W_2$ with $\End(W_1) = \Start(W_2)$,
     we define the {\em concatenation\/} $W = W_1 \cdot W_2$ by setting
     $\Start(W) = \Start(W_1)$, $\End(W) = \End(W_2)$, and $\Edges(W)$ to be
     the juxtaposition of the sequences $\Edges(W_1)$ and $\Edges(W_2)$.
If $W=W_1\cdot C\cdot W_2$ where~$C$ is a closed walk, then~$W'=W_1\cdot W_2$
is also a walk with the same start and end nodes as~$W$.

A walk is a {\em path\/} if it is non-closed and does not contain a nonempty
closed walk as a subwalk.
A closed walk is a {\em cycle\/} if it does not contain a nonempty
closed walk as a proper subwalk.
As cycles can be empty, there is a cycle of length~$0$ at each node of~$G$.

If~$i$ and~$j$ are two nodes of~$G$, let $\Pa_G(i,j)$ denote the set of walks~$W$ in graph~$G$ with $\Start(W)=i$ and $\End(W)=j$,
and $\Pa_G(\ito)$ the set of walks~$W$ in~$G$ with $\Start(W)=i$. If~$n$ is a
nonnegative integer, we write $\Pa_G^n(i,j)$ (respectively $\Pa_G^n(\ito)$) for the
set of walks in $\Pa_G(i,j)$ (respectively $\Pa_G(\ito)$) of length~$n$.
When no confusion can arise, we will omit the subscript $G$.

A graph $G'=(V',E')$ is a {\em subgraph\/} of~$G$ if $V'\subseteq V$
     and $E' \subseteq E$.
For a nonempty subset~$E'$ of~$E$, let the {\em subgraph  of $G$ induced by edge set $E'$\/} be the
     graph~$(V',E')$ where $V' = \{i
     \in V \mid \exists j\in V: (i,j)\in E' \vee (j,i)\in E'\}$.
A graph~$G$ is {\em strongly connected\/} if, for all nodes~$i$ and~$j$ in~$G$,
there exists a walk from~$i$ to~$j$.
A subgraph $H$ of $G$ is a {\em strongly connected
     component\/} of $G$ if $H$ is maximal with respect to the
     subgraph relation such that $H$ is strongly connected.

The {\em girth\/}~$g(G)$ of a graph~$G$ is the minimum length of a nonempty
cycle in~$G$.
For a strongly connected graph~$G$, its {\em cyclicity\/}~$\gamma(G)$ is the greatest
common divisor of cycle lengths in~$G$.
If~$G$ is not strongly connected, then its cyclicity~$\gamma(G)$ is equal to the least common
multiple of the cyclicities of its strongly connected components.

\subsection{Linear max-plus systems}

Let $\IRmax = \IR \cup \{-\infty\}$.
In this paper, we follow the convention $\max\emptyset=-\infty$.

A matrix with entries in $\IRmax$ is called a {\em max-plus matrix}.
If~$A$ is an $M\times N$ max-plus matrix and~$B$ is an $N\times Q$ max-plus
matrix, then the {\em max-plus product\/}~$A\otimes B$ is an $M\times Q$
max-plus matrix defined by 
\[(A\otimes
     B)_{i,j} = \max_{1\leqslant k\leqslant N} \big( A_{i,k} + B_{k,j} \big)
     \enspace.\]
If~$A$ is an $N\times N$ max-plus matrix and~$n$ is a nonnegative integer, we
denote by~$A^{\otimes n}$ the $n$ times iterated matrix product of~$A$.
That is, $(A^{\otimes 0})_{i,i}=0$ and $(A^{\otimes 0})_{i,j}=-\infty$
if~$i\neq j$, and $A^{\otimes n} = A\otimes A^{\otimes (n-1)}$ if $n\geqslant1$.
Given a column vector $v \in \IRmax^N$, the corresponding {\em
     linear max-plus system\/} is the sequence of vectors~$x(n)$ defined by     
\begin{align}
x(n) = \begin{cases}
               v & \text{if } n = 0\\
               A \otimes x(n-1) & \text{if } n \geqslant1\enspace.
       \end{cases}\label{eq:x}
\end{align}
Clearly $x(n) = A^{\otimes n}\otimes v$.
Let  $x= \langle A,v\rangle$, i.e., $\langle A,v\rangle$ denotes 
	the $A$-linear system with the initial vector~$v$.

To an $N\times N$ max-plus matrix~$A$ naturally corresponds a graph~$G(A)$ 
     with set of nodes $\{1,\dots,N\}$ containing an edge
     $(i,j)$ if and only if $A_{i,j}$ is finite.
The matrix~$A$ is said to be {\em irreducible\/} if~$G(A)$ is strongly connected.

We refer to~$A_{i,j}$ as the {\em $A$-weight\/} of edge~$(i,j)$ in~$G(A)$.
If~$W$ is a walk in~$G(A)$, we abuse notation by writing~$A(W)$ for the weight
of walk~$W$, i.e., the sum of the weights of its edges. We follow the
convention that the value of the empty sum is zero, i.e., $A(W)=0$ if~$W$ is an
empty walk.
Given a column vector $v \in \IRmax^N$, we write~$A_v(W) = A(W) + v_j$
where $j=\End(W)$ for~$W$'s {\em $A_v$-weight}.
From these definitions, one can easily establish the following correspondence between 
	the matrix power $A^{\otimes n}$ (respectively the vector $A^{\otimes n}\otimes v$) 
	and the weights of some walks in $G(A)$.

 \begin{proposition}
Let~$i$ and $j$ be two nodes of~$G(A)$, and let~$n$ be a nonnegative integer. 
Then the following equations hold
\begin{align}
  (A^{\otimes n})_{i,j} &= \max\left\{ A(W) \mid W\in\Pa_{G(A)}^n(i,j)
\right\}\notag\\
  (A^{\otimes n}\otimes v)_i &= \max\left\{ A_v(W) \mid W\in\Pa_{G(A)}^n(\ito) \right\}\enspace.
\notag
\end{align}
\end{proposition}

\subsection{The critical subgraph}

A nonempty closed walk~$C$ in~$G(A)$ is said to be {\em critical\/} if its average
	$A$-weight~$A(C)/\ell(C)$ is maximal, i.e., if it is equal to
\[ \lambda(A) = \max \big\{ A(C)/\ell(C) \mid C \text{ is a nonempty closed
walk in }G(A)  \big\} \enspace, \]
which is easily seen to be finite whenever there is at least one cycle in~$G(A)$.
A node of $G(A)$ is {\em critical\/} if it is a node of a critical closed walk
     in~$G(A)$, and an edge of~$G(A)$ is {\em critical\/} if it is an edge
     of a critical closed walk in~$G(A)$.
The {\em critical subgraph\/} of~$G(A)$, denoted by~$G_\cc(A)$, is
     the subgraph of~$G(A)$ induced by the set of critical edges
     of~$G(A)$.
We recall a useful property of closed walks in $G_\cc(A)$ (for instance
see~\cite[Lemma~2.6]{workplus} for a proof).

\begin{proposition}\label{prop:paths:in:crit:comps:are:critical}
Every nonempty closed walk in $G_\cc(A)$ is critical in $G(A)$.
\end{proposition}
Let us denote $\gamma(A) = \gamma\big(G_\cc(A)\big)$.
	
\subsection{Eventually periodic sequences}\label{subsec:period}

Let~$I$ be an arbitrary nonempty set and $f:\IN\to\IRmax^I$.
Further let $\pi$ be a positive integer and $\varrho\in\IR$.
The sequence $f$ is {\em eventually periodic with period $\pi$ and ratio $\varrho$}
     if there exists a nonnegative integer $T$ such that
\begin{equation}\label{eq:def:periodic}
\forall i\in I\,:\ \forall n\geqslant T\,:\ f_i(n+\pi) = f_i(n)+ \pi\cdot \varrho \enspace.
\end{equation}%
We call such a~$T$ a {\em transient\/} of~$f$ with respect to~$\pi$
and~$\rho$.
The ratio is unique if not all component-wise sequences $\big(f_i(n)\big)_n$ are
eventually constantly equal to $-\infty$.
In all cases, the set of transients of~$f$ is independent of the ratio.

Obviously if~$\sigma$ is any multiple of~$\pi$, then~$f$ is also eventually periodic with
period~$\sigma$
	and ratio $\varrho$.
Hence,  there always exists a common period of two eventually periodic
	sequences.

For every period~$\pi$, there exists a unique minimal transient~$T_\pi$.
The next lemma shows that these minimal transients do, in
     fact, not depend on~$\pi$. We will henceforth call this common value {\em the\/}
	     transient of~$f$.

\begin{proposition}\label{prop:period}
Let~$\pi$ and~$\sigma$ be two periods of an eventually periodic sequence~$f$ with respective minimal
     transients~$T_\pi$ and~$T_\sigma$.
Then~$T_\pi=T_\sigma$.
\end{proposition}
\begin{proof}
Denote by ${\mathbf \Pi}_f$ the set of periods of~$f$ and by~$\varrho$ a ratio
of~$f$.
Clearly, ${\mathbf \Pi}_f$ is a nonempty subset of $\IN^*$ closed under
     addition.
Let  $\pi_0 = \min {\mathbf \Pi}_f$ be the minimal period of~$f$; hence
     $\pi_0\IN^* \subseteq {\mathbf \Pi}_f$.
Denote by~$T_0$ the minimal transient with respect to period $\pi_0\in {\mathbf
\Pi}_f$.
Let $\pi= a\pi_0 +b$ be the Euclidean division of~$\pi$ by~$\pi_0$.
For any integer $n\geqslant \max \{T_\pi ,T_0 -b\}$, 
	\[f(n+\pi)=f(n) + \pi\varrho = f(n + b) +a\pi_0\varrho\enspace.\]
It follows that either $b=0$ or $b$ is a period of $f$.
Since $b\leqslant \pi_0-1$ and $\pi_0$ is the smallest period of~$f$, we have $b=0$,
	i.e., $\pi_0$ divides $\pi$.
We have thus shown ${\mathbf \Pi}_f\subseteq \pi_0 \IN^*$ and thus ${\mathbf
\Pi}_f = \pi_0 \IN^*$.
Hence $\pi=a\pi_0$ for some positive integer~$a$.

Since for any $n \geqslant T_0$,
	$f(n + a \pi_0) = f(n) + a \pi_0 \varrho,$
	we have $T_\pi \leqslant T_0$.
We now prove that  $T_\pi = T_0$ by induction on $a$.
\begin{enumerate}
	\item The base case $a=1$ is trivial.
	\item Let~$a\geqslant2$.
	Denote by $T'$ the minimal transient with respect to period~$(a-1)\pi_0$.
	By the inductive hypothesis, $T' = T_0$.
	For any integer $n \geqslant T_\pi$, 
	\begin{equation}
	f(n + a \pi_0) = f(n) + a \pi_0 \varrho\enspace.
	\end{equation}
	Moreover, if $n + \pi_0 \geqslant T'$ then 
	\begin{equation}
	f(n + a  \pi_0) = f(n + \pi_0) + (a-1)\pi_0 \varrho\enspace.
	\end{equation}
	It follows that for any integer $n \geqslant \max \{ T' - \pi_0, T_\pi \}$, 
	\begin{equation}
	f(n + \pi_0) = f(n) + \pi_0\varrho\enspace.
	\end{equation}
	Hence $T_0 \leqslant \max \{ T' - \pi_0, T_\pi \}$, and by inductive assumption 
	$T_0 \leqslant \max \{ T_0 - \pi_0, T_\pi \}$.
	We derive $T_0 \leqslant T_\pi$, and so $T_0 = T_\pi$, which concludes the proof. 
\qedhere
\end{enumerate}	
\end{proof}

Cohen et al.\ proved eventual periodicity of  irreducible max-plus matrix powers in the
following analog of the Perron-Frobenius theorem in classical linear algebra.

\begin{theorem}[Cyclicity Theorem~{\cite{CDQV83}}]\label{thm:cyclicity}
If~$A$ is irreducible, then the sequence of matrix powers~$A^{\otimes n}$
	is eventually periodic with period~$\gamma(A)$ and 
	ratio~$\lambda(A)$.

Consequently, every linear max-plus system with an irreducible matrix~$A$ is
	eventually periodic with period~$\gamma(A)$ and
	ratio~$\lambda(A)$.
\end{theorem}

We call the transient of the sequence of matrix powers~$A^{\otimes n}$ the 
	{\em transient of matrix~$A$}, and the transient of the sequence of vectors~$A^{\otimes n}\otimes v$ 
	the {\em transient of the system~$\langle A,v\rangle$}.

For any  $\mu \in \IR$, let $A + \mu$ denote the matrix obtained by 
	adding $\mu$ to each entry of~$A$.
Since $(A + \mu)^{\otimes n} = A^{\otimes n} + n\mu$, we easily check that 
	$ G_\cc(A + \mu) = G_\cc(A)$, $\lambda(A + \mu) = \lambda(A) + \mu$,
	and the matrix transients of~$A$ and $A + \mu$
	(resp.\ the system transients of~$\langle A,v\rangle$ and $\langle A +
\mu,v\rangle$) are equal.

\section{Strategy Outline}\label{sec:bounds:outline}

This section describes our graph-based strategy to prove upper bounds on the transient
	of the system~$\langle A,v\rangle$, given an irreducible $N\times N$ matrix~$A$ and a vector~$v\in\IR^N$. 
We also explain how a slight modification of this strategy provides upper bounds on the transient 
	of $A$. 

We start by defining for a set~${\mathbf N}$ of nonnegative integers 
	and a node~$i$, an ${\mathbf N}$-{\em realizer\/} for node~$i$ to be any walk of
	maximum $A_v$-weight in the set of walks in $\Pa(\ito)$
	with length in ${\mathbf N}$. 
As shown in the next proposition, of particular interest is the case of sets ${\mathbf N}$ of the form
	$$\realrem{B}{n,\pi} = 
	\{ m\in \IN \mid m \geqslant B\  \wedge \ m\equiv n\pmod \pi \}$$	
	where $B$, $n$, and $\pi$ are positive integers.
	
\begin{proposition}\label{prop:final:step}
Let~$B$ and~$\pi$ be positive integers.
If there exists, for every node~$i$ and every integer~$n\geqslant B$, an
	$\realrem{B}{n,\pi}$-realizer for~$i$ of length~$n$, then~$B$ is an 
	upper bound on the system transient.
\end{proposition}
\begin{proof}
Let~$i$ be a node.
For each integer $n\geqslant B$, let $W_n$ be an $\realrem{B}{n,\pi}$-realizer 
		for~$i$ of length $n$.
Denote by~$X(n)$ the set of walks~$W$ in~$\Pa(\ito)$ with
	$\ell(W)\in \realrem{B}{n,\pi}$, and
	    let~$x(n)$ be the maximum of values~$A_v(W)$ where~$W\in
	    X(n)$.
It is $x(n) = A_v(W_n)$.

From $n+\pi \equiv n\pmod \pi$ follows $X(n+\pi) = X(n)$ and so $x(n+\pi)=x(n)$.
Moreover, we have~$\Pa^n(\ito) \subseteq X(n)$ and $\Pa^{n+\pi}(\ito)\subseteq X(n+\pi)$,
	    which implies $(A^{\otimes n}\otimes v)_i \leqslant x(n)$ and 
		$(A^{\otimes (n+\pi)}\otimes v)_i \leqslant x(n+\pi)$. 
Conversely because~$W_n\in \Pa^n(\ito)$, we 
	have~$(A^{\otimes n}\otimes v)_i \geqslant A(W_n)=x(n)$.
Similarly,~$(A^{\otimes (n+\pi)}\otimes v)_i \geqslant A(W_{n+\pi})=x(n+\pi)$. 
Since $x(n+\pi)=x(n)$, it follows that $(A^{\otimes n}\otimes v)_i = (A^{\otimes (n+\pi)}\otimes v)_i$.
Noting Proposition~\ref{prop:period} now concludes the proof.
\end{proof}

Based on Proposition~\ref{prop:final:step}, we now define a strategy for determining 
	upper bounds on system transients.
Let~$n$ be a nonnegative integer and~$i$ be a node.
Denote by~$\pi$ the least common multiple of cycle lengths in the critical
subgraph~$G_\cc$.
Note that~$\pi$ is a multiple of~$\gamma=\gamma(A)$.
The strategy includes an additional parameter~$B$ to be
chosen in step~4.

\medskip

\begin{enumerate}
\item {\em Normalized matrix.}	Because the transients of~$A$ and of~$\overline{A} = A-\lambda(A)$
	are equal, and~$\lambda(\overline{A})=0$, we can reduce the general
	case to the case~$\lambda(A)=0$.
The condition~$\lambda(A)=0$ guarantees the existence of realizers for every
	nonempty~$\mathbf{N}\subseteq\IN$ and yields that adding critical cycles to a walk does
	not change its $A$-weight.
The rest of the strategy hence considers an irreducible matrix~$A$ such that $\lambda(A) =0$.
Let $W$ be an $\realrem{B}{n,\pi}$-realizer for node~$i$.
\item {\em Critical bound.}	We show that for~$B$ large enough, i.e., $B$ greater or equal to 
	some {\em critical bound}~$B_\cc$, the realizer $W$ contains at least one critical node~$k$.
\item {\em Walk reduction.}	Next we show that for every divisor $d$ of $\pi$, by removing subcycles, 
	we can {\em reduce}~$W$ to a new walk~$\hat{W}$ which starts at node~$i$, contains the critical node~$k$,
	whose length $\ell(\hat{W})$ is in the same residue class modulo~$d$ as~$\ell(W)$, and 
	$\ell(\hat{W})$ is upper-bounded by a term linear in the number of nodes in the graph.
\item {\em Pumping in the critical graph.}	Since~$d$ divides~$\pi$,~$d$ divides~$n-\ell(\hat{W})$,
	and for two appropriate choices of $d$ and for~$n$ sufficiently large ($n \geq B_d$), we show 
	how to complete $\hat{W}$ by adding to it a critical closed walk starting from~$k$ in order to 
	obtain a new walk of length~$n$ starting at node~$i$.

For $B= \max\{ B_\cc, B_d\}$, this yields an $\realrem{B}{n,\pi}$-realizer of length~$n$, 
	because removing cycles at most increases the weight and adding a critical closed path 
	does not change the weight. 
Proposition~\ref{prop:final:step} then shows that~$B$ is a bound on the transient.
\end{enumerate}	

For the transient of the matrix~$A$, we can follow a similar strategy: 
	we consider~$\Pa(i,j)$ instead of~~$\Pa(\ito)$, and for a set~${\mathbf N}$ of 
	nonnegative integers we define an ${\mathbf N}$-{\em realizer\/} for the pair of
	nodes~$i,j$ to be any walk of maximum $A$-weight in the set of walks $\Pa(i,j)$
	with length in ${\mathbf N}$.
As for walks in $\Pa(\ito)$, we can show that any walk of maximum  $A$-weight in $\Pa(i,j)$  with 
	length in $\realrem{B}{n,\pi}$ contains at least one critical node
	if $B$ is greater or equal to some critical bound~$B'_\cc$.
Since the walk reduction described above actually preserves both the starting and ending nodes,
	then we can derive an upper-bound on the transient of~$A$.
In fact, we will not develop this parallel strategy for matrices, but 
	we rather propose a different method, which consists in computing
	a bound on the transient of matrix~$A$ from our bounds on transients of some specific
systems~$\langle A,v\rangle$.

\section{Critical Bound}\label{sec:criticalbound}

In this section, we carry out step~2 of our strategy.
More precisely, we prove that any walk of maximum $A_v$-weight in the set of walks $\Pa^n(\ito)$
	necessarily contains a critical node if $n$ is large enough.
	
Let $A$ be an $N\times N$ max-plus matrix, and assume $A$ is irreducible.
We write~$\lambda$ for $\lambda(A)$, $\lambda_\nc$ for the maximum average $A$-weight of closed walks
without critical nodes,~$\delta$ for the minimum $A$-weight,~$\Delta$ for the maximum $A$-weight,~$\Delta_\nc$ for the 
	maximum $A$-weight of edges between non-critical nodes, and~$\lVert v\rVert$ for the difference
	of the maximum and minimum entry of vector~$v$.
We assume~$\lVert v\rVert$ to be finite until Section~\ref{sec:matrix}, in
which we generalize our results to arbitrary~$v$.
By comparing the possible $A_v$-weights of walks that do and do not visit~$G_\cc$, we
	can derive an explicit critical bound~$B_\cc$, which holds for arbitrary~$\lambda$.
\begin{proposition}[Critical Bound]\label{prop:n:zero}
Each walk with maximum $A_v$-weight in $\Pa^n(\ito)$ contains a critical node if
	$n\ge B_\cc$ where 
	\begin{equation}
	B_\cc = \max\left\{ N\ ,\ \frac{\lVert v\rVert +
(\Delta_{\nc}-\delta)\,(N-1)}{\lambda-\lambda_\nc}\right\} \enspace. \notag
	\end{equation}
\end{proposition}
\begin{proof}
We first reduce to the case $\lambda=0$.
Let~$\overline{A}$ be the normalized matrix $\overline{A} = A - \lambda$.
The parameters $\overline{\delta}$, $\overline{\Delta}_{nc}$, and $\overline{\lambda}_{nc}$ 
	for the matrix~$\overline{A}$ are obtained by subtracting $\lambda$ from the respective parameters 
	of~$A$.
Hence $\overline{\lambda}=0$, and a walk is of maximum $A_v$-weight in $G(A)$ if and only if
     it is a walk of maximum $\overline{A}_v$-weight in~$G(\overline{A})=G(A)$.
The term 	$\frac{\lVert v\rVert + (\Delta_{\nc}-\delta)\,(N-1)}{\lambda-\lambda_\nc}$ 
	should hence be substituted by 
	$\frac{\lVert v\rVert + (\overline{\Delta}_{\nc}-\overline{\delta})\,(N-1)}{-\overline{\lambda_\nc}}$ 
	when considering~$\overline{A}$ instead of~$A$, and we can assume $\lambda=0$ in the rest of the proof.

If $\lambda_\nc=-\infty$, then every nonempty cycle contains a critical node. 
Because every walk of length greater or equal to~$N$ necessarily contains a cycle 
	as a subwalk and because $B_\cc\geqslant N$, in particular every walk with maximum
	$A_v$-weight in $\Pa^n(\ito)$ contains a critical node if $n\geqslant B_\cc$ and
	$\lambda_\nc=-\infty$. 

We now consider the case $\lambda_\nc\neq-\infty$.
We proceed by contradiction: Suppose  that there exists an integer $n$ such that 
	$n \ge B_\cc$, a node $i$ and 
	a walk of maximum $A_v$-weight in $\Pa^n(\ito)$ with non-critical nodes only;
	let~$\hat{W}$ be such a walk.
Let~$W_0$ be the {\em acyclic part\/} of~$\hat{W}$, defined in the following
manner:
Starting at~$\hat{W}$, we repeatedly remove nonempty subcycles from the walk until
	we arrive at a path.
In general there are several possible choices of which subcycles to remove, but we
	fix some global choice function to make the construction of $W_0$ deterministic.

Next choose a critical node~$k$, and then a prefix~$W_\cc$ of~$W_0$, such
     that the distance between~$k$ and the end node of $W_\cc$ is minimal.
Let~$W_2$ be a path of minimal length from the end node of~$W_\cc$ to~$k$.
Let $W_3$ be the walk such that $W_0 = W_\cc \cdot
     W_3$.
Further let~$C$ be a critical cycle starting at $k$.

We distinguish two cases for~$n$, namely (a) $n \ge
     \ell(W_\cc)+\ell(W_2)$, and (b) $n < \ell(W_\cc)+\ell(W_2)$.

\medskip 

{\em Case a:} Let~$m \in \IN$ be the quotient in the Euclidean division of
     $n - \ell(W_\cc) - \ell(W_2)$ by $\ell(C)$, and choose~$W_1$ to be a prefix 
	of~$C$ of length $n - \big(\ell(W_\cc)+\ell(W_2)+m\cdot \ell(C)\big)$
	(cf. Figure~\ref{fig:thm:n:zero}).
Clearly $W_1$ starts at~$k$.
If we set $W = W_\cc\cdot W_2\cdot C^m\cdot W_1$, we get
     $\ell(W)=n$ and 
	\begin{equation}\label{eq:crit:lower:bound}
		A_v(W) \geqslant \min_{1\leqslant j\leqslant N}(v_j) + A(W_\cc)+A(W_2)+A(W_1)
	\end{equation}
	since we assume $\lambda=0$.  

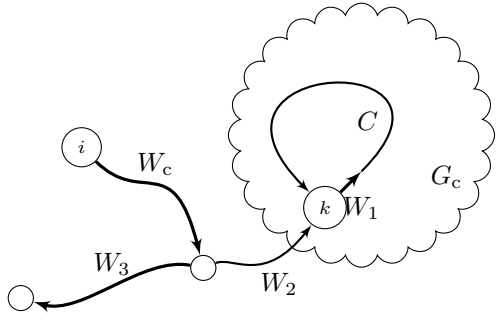
\begin{figure}[hbt]
\centering
\begin{tikzpicture}[>=latex',scale=0.8]
	\node[shape=circle,draw] (i) at (-2,2) {$\scriptstyle i$};
	\node[shape=circle,draw] (k) at (0,0) {};
	\node[shape=circle,draw] (j) at (2,1) {$\scriptstyle k$};
	\node[shape=circle,draw] (end) at (-3,-.5) {};
	\draw[very thick,->] (i) .. controls (-1,1) and (-.5,2)  ..
node[midway,above]{${W}_\cc$}  (k);
	\draw[thick,->] (k) .. controls (0.5,0.2) and (1,-0.3)  .. node[near end,below=1mm]{${W}_2$}  (j);
	\draw[very thick,->] (j) --  (2.6,1.6) node[below=2mm] {$W_1$};
	\draw[thick,->] (2.6,1.6) .. controls +(2.1,2.1) and +(-2.5,2.5)  .. node[near start,below left]{$C$}  (j);
	\draw[very thick,->] (k) .. controls +(-1,0.2) and +(1,-0.3) .. node[above]{$W_3$} (end);
	\node[cloud, cloud puffs=24, draw,minimum width=3.5cm, minimum height=3.5cm] at (2.6,2.2) {};
	\node at (4,1.5) {$G_\cc$};
\end{tikzpicture}
\caption{Walk~$W$ in proof of Proposition~\ref{prop:n:zero}}
\label{fig:thm:n:zero}
\end{figure}

For the $A_v$-weight of $\hat{W}$, we have
	\begin{align}
		A_v(\hat{W}) \leqslant  A_v(W_0) +
        \lambda_{\nc} \cdot \bigr(\ell(\hat{W}) - \ell(W_0) \bigr)
 		\leqslant \max_{1\leqslant j\leqslant N}(v_j) + A(W_0)+
        \lambda_{\nc}\cdot \bigr(\ell(\hat{W}) - \ell(W_0) \bigr)\label{eq:upper:bound}
	\end{align}
By assumption $A_v(\hat{W}) \ge A_v(W)$, and from
     \eqref{eq:crit:lower:bound}, \eqref{eq:upper:bound}, and $\lambda_{\nc} < 0$ 
     we therefore obtain
\begin{align}
\ell(\hat{W}) \le \frac{\lVert v \rVert + A(W_3)-A(W_1)-A(W_2)}{-\lambda_{\nc}}+
                 \ell(W_0)
 \le \frac{\lVert v \rVert +
\Delta_{\nc}\,\ell(W_3)-\delta\,(\ell(W_1)+\ell(W_2))}{-\lambda_{\nc}}+
                 \ell(W_0)\label{eq:the:bound}
\end{align}%
Denote by $N_{\nc}$ the number of non-critical nodes.
The following three inequalities trivially hold:  $\ell(W_3) \le N_{\nc}-1$, 
	$\lambda_{\nc} \ge \delta$,  and $\ell(W_1) < N-N_{\nc}$.
Since there is at least one critical node, we have $\ell(W_3) < N-1$.
Moreover from the minimality constraint for the length of~$W_2$ follows that
 	$\ell(W_2)+\ell(W_0) \le N_{\nc}$.
Thereby
\begin{align}
\ell(\hat{W})  < \frac{\lVert v\rVert + (\Delta_{\nc}-
   \delta)\,(N-1)}{-\lambda_{\nc}}\enspace,\label{eq:bound:2}
\end{align}%
	a contradiction to $n \ge B_\cc$.
The lemma follows for case a.

\medskip

{\em Case b:} In this case $\ell(W_c) \le n < \ell(W_c) + \ell(W_2)$,
	and we set $W = W_\cc\cdot W'_2$, where
     $W'_2$ is a prefix of $W_2$, such that $\ell(W)=n$.
Hence,
	\begin{equation}\label{eq:crit:lower:bound_caseb}
		A_v(W) \geqslant \min_{1\leqslant j\leqslant N}(v_j) + A(W_\cc)+A(W'_2)\enspace.
	\end{equation}
We again obtain~\eqref{eq:upper:bound}.
By assumption $A_v(\hat{W}) \ge A_v(W)$, and by similar arguments as in case a 
	we derive
\begin{align}
\ell(\hat{W}) &\le \frac{\lVert v \rVert + A(W_3)-A(W'_2)}{-\lambda_{\nc}}+
                 \ell(W_0)\enspace\notag\\
\intertext{and since $W'_2$ is a prefix of $W_2$ with $\ell(W'_2) < \ell(W_2)$,}
\ell(\hat{W}) &< \frac{\lVert v \rVert +
\Delta_{\nc}\,\ell(W_3)-\delta\,\ell(W_2)}{-\lambda_{\nc}}+
                 \ell(W_0)\enspace,\notag
\end{align}%
	which is less or equal to the bound obtained in \eqref{eq:the:bound} of case~a.
By similar arguments as in case~a, the lemma follows in case~b	.
\end{proof}
	
In case~$A$ is an integer matrix, i.e., all finite entries of $A$ are
     integers, the term~$\lambda - \lambda_\nc$ cannot become
     arbitrarily small: This is obvious when $\lambda_\nc = -\infty$;
     otherwise, let $C_0$  be a critical cycle, and let $C_1$ be a
     cycle such that $\lambda_{\nc}= A(C_1)/\ell(C_1)$.
Then we have 
\[\lambda - \lambda_{\nc} = \frac{A(C_0)\ell(C_1) - A(C_1)\ell(C_0)}{\ell(C_0)\ell(C_1)} \enspace,\]
 	and so
	\begin{equation}
		\frac{1}{\lambda - \lambda_\nc} \leqslant  
		(N-N_\nc)\cdot N_\nc \le \frac{N^2}{4} \enspace,\label{eq:int_bound}
	\end{equation}%
where $N_{\nc}$ denotes the number of non-critical nodes.
It follows that, in case of integer matrices, the critical bound $B_\cc$ is
in $O(\lVert A \rVert \cdot N^3)$ for a given initial vector.

\section{Walk Reduction}\label{sec:red}

This section concerns step 3 of our strategy and constitutes its core. 
Given a walk~$W$, a positive integer~$d$, and a node~$k$ of~$W$, we define a
reduced walk,  denoted $\Red_{d,k}(W)$, such that
(a) it contains node~$k$ and has the same start and end nodes as~$W$,
(b) its length is in the same residue class modulo~$d$ as~$W$'s length, and 
(c) its length is bounded by $(d-1)+ 2d\,(N-1)$.

Properties (a) and (b) can be achieved by removing a collection of cycles from~$W$ whose
	combined length is divisible by~$d$ and whose removal retains connectivity
	to~$k$.
The key point of the reduction is that we can iterate this cycle removal until
	the resulting length is at most $(d-1)+ 2d\,(N-1)$.

We call a finite, possibly empty, sequence of nonempty subcycles~$\mathcal{S} = (C_1,C_2,\dots,C_n)$  
 	a {\em cycle pattern of a walk}~$W$ if there exist walks~$U_0,U_1,\dots,U_n$ such that 
\begin{equation}\label{eq:disjointness}
W = U_0\cdot C_1\cdot U_1 \cdot C_2 \cdots U_{n-1} \cdot C_n\cdot U_n\enspace.
\end{equation}
The choice of the~$U_m$'s in~\eqref{eq:disjointness} may be not unique,
	and we fix some global choice function to make it deterministic.
Then we define {\em the removal of}~$\mathcal{S}$ {\em from}~$W$ as
     $$ \Rem(W,\mathcal{S}) = U_0\cdot U_1\cdots U_n \enspace.$$ 
The walks~$W$ and~$\Rem(W,\mathcal{S})$ have the same start and end nodes.
Furthermore $\ell\big(\!\Rem(W,\mathcal{S})\big) = \ell(W) - \ell(\mathcal{S})$
      where $\ell(\mathcal{S}) = \sum_{C\in\mathcal{S}}\ell(C)$.
In particular,~$\Rem(W,\mathcal{S})=W$ if and only if~$\ell(\mathcal{S})=0$,
	i.e., $\mathcal{S}$ is the empty cycle pattern.

Given any node $k$ of a walk~$W$, let~${\mathbf S}_k(W)$ denote the set of
    cycle pattern~$\mathcal{S}$ of~$W$ whose removal does not impair connectivity to~$k$, 
	i.e.,~$k$ is a node of~$\Rem(W,\mathcal{S})$.
Further for any positive integer $d$, define ${\mathbf S}_{d,k}(W)$ as the subset of 
	cycle pattern  $\mathcal{S}\in{\mathbf S}_k(W)$ that, in addition, leave the length's 
	residue class modulo~$d$ intact, i.e., $\ell(\mathcal{S}) \equiv 0 \pmod d$.
The set~${\mathbf S}_{d,k}(W)$ is not empty, because~$k$ is a node of~$W$ and we can 
	hence choose~$\mathcal{S}$ to be the empty cycle pattern.

Choose~$\mathcal{S}\in {\mathbf S}_{d,k}(W)$ such that~$\ell(\mathcal{S})$ is maximal.
There may be several possible choices for~$\mathcal{S}$,  
	and we again fix some global choice function to make the choice deterministic;
	then set $$\Step_{d,k}(W)=\Rem(W,\mathcal{S})\enspace.$$
The limit 	$$\Red_{d,k}(W) = \lim_{t\to\infty} \Step_{d,k}^t(W)$$ exists because the 
	sequence of walks $(\Step_{d,k}^t(W))_{t\geq 0}$ is stationary after at
	most~$\ell(W)$ steps, and we call it {\em the $(d,k)$-reduction of}~$W$.
More specifically,~$\Red_{d,k}(W) = W$ if and only if $\mathbf{S}_{d,k}(W)$ is reduced
	to the sole empty cycle pattern.
The walks $W$ and $\Red_{d,k}(W)$ have the same start and end nodes.
Also,~$k$ is a node of~$\Red_{d,k}(W)$ and
$ \ell\big(\! \Red_{d,k}(W) \big) \equiv \ell(W) \pmod d$. 

Bounding the length of $\Red_{d,k}(W)$ relies on a simple arithmetic lemma
	which is an elementary application of the pigeonhole principle:

\begin{lemma}\label{lem:erdos}
Let $d$ be a positive integer and let $x_1,\dots,x_d\in\mathds{Z}$.
Then there exists a nonempty set $I\subseteq\{1,\dots,d\}$ such that  
$\displaystyle \sum_{i\in I} x_i \equiv 0 \pmod d$.
\end{lemma}%

\begin{theorem}\label{thm:red:upper:bound}
For each positive integer~$d$ and each node~$k$, 
	the length of the $(d,k)$-reduction of any walk~$W$ containing node~$k$
	is at most equal to $(d-1) + 2d \cdot (N-1)$:
	$$\ell\big(\!\Red_{d,k}(W) \big)  \leqslant (d-1) + 2d \cdot (N-1) \enspace.$$
\end{theorem}

\begin{proof}
We denote~$\hat{W} = \Red_{d,k}(W)$.
By definition of the $(d,k)$-reduction, $\Red_{d,k}(\hat{W})=\hat{W}$. 
Let $\mathcal{S}$ be any cycle pattern of~$\hat{W}$ in~$\mathbf{S}_k(\hat{W})$,
	and let $n$ be the number of cycles of $\mathcal{S}$.
We first show that $n \leqslant d-1$.
Indeed, suppose for contradiction that $n\geqslant d$.
Then Lemma~\ref{lem:erdos} implies that
	there exists a nonempty subsequence of~$\mathcal{S}$ that is in ${\mathbf S}_{d,k}(\hat{W})$,
	which contradicts $\Red_{d,k}(\hat{W})=\hat{W}$.
	
Now let us choose $\mathcal{S}$ in~$\mathbf{S}_k(\hat{W})$ with maximal~$\ell(\mathcal{S})$.
If $\mathcal{S} = (C_1,C_2,\dots,C_n)$, then  there exist walks~$U_0,U_1,\dots,U_n$ such that 
	\begin{equation}
		\hat{W} = U_0\cdot C_1\cdot U_1 \cdot C_2 \cdots U_{n-1} \cdot C_n\cdot U_n\enspace.
		\notag
	\end{equation}

\begin{figure}[tbh]
\centering
\begin{tikzpicture}[>=latex',scale=1.0]
	\node[shape=circle,draw] (i) at (-6,0) {};
	\node[shape=circle,draw] (k) at (0,.8) {$\scriptstyle k$};
	\node[shape=circle,draw] (j) at (6,0) {};
	\node[shape=circle,draw] (n1) at (4,0) {};
	\node[shape=circle,draw] (n2) at (2,0) {};
	\node[shape=circle,draw] (n3) at (-2,0) {};
	\node[shape=circle,draw] (n4) at (-4,0) {};
	\node  at (-0,-0.3) {$U_r$};
	\draw[thick,->] (i) -- node[below] {$U_0$} (n4);
	\draw[thick,->,dotted] (n4) -- node[below] {} (n3);
	\draw[thick,->] (n3) -- node[above] {$W_1$} (k);
	\draw[thick,->] (k) -- node[above] {$W_2$} (n2);
	\draw[thick,->,dotted] (n2) -- node[below] {} (n1);
	\draw[thick,->] (n1) -- node[below] {$U_n$} (j);
	\draw[thick,->] (n4) .. controls +(1,1) and +(-1,1) .. node[above] {$C_1$} (n4);
	\draw[thick,->] (n3) .. controls +(1,1) and +(-1,1) .. node[above] {$C_r$} (n3);
	\draw[thick,->] (n2) .. controls +(1,1) and +(-1,1) .. node[above] {$C_{r+1}$} (n2);
	\draw[thick,->] (n1) .. controls +(1,1) and +(-1,1) .. node[above] {$C_n$} (n1);
\end{tikzpicture}
\caption{Structure of the reduced walk $\hat{W} = \Red_{d,k}(W)$}
\label{fig:lem:red:upper:bound}
\end{figure}
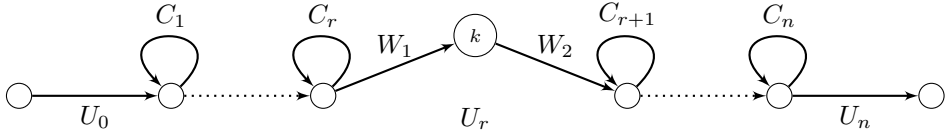
By definition of~$\mathbf{S}_k(\hat{W})$,~$k$ is a node of~$\Rem(\hat{W},\mathcal{S})$.
Hence there exists some index~$r$ such that~$k$ is a node of~$U_r$.
Each~$U_m$ with~$m\neq r$ is a (possibly empty) path, because otherwise we could
	add a nonempty subcycle of~$U_m$ to~$\mathcal{S}$, a contradiction to the maximality 
	of~$\ell(\mathcal{S})$.
Similarly, if~$U_r=W_1\cdot W_2$ such that~$k$ is the end node of~$W_1$, then
	both~$W_1$ and~$W_2$ are (possibly empty) paths.
Hence, apart from the at most $(d-1)$ cycles in~$\mathcal{S}$, the reduced
	walk~$\hat{W}$ consists of at most~$(d+1)$ subpaths.
Noting that each cycle has length at most~$N$ and each path has length at most~$(N-1)$
	concludes the proof.
\end{proof}

\section{Exploration Penalty}\label{sec:explorationpenalty}

One of the two pumping techniques that we develop in step~4 of our strategy for the
	construction of arbitrarily long closed walks in the critical graph
$G_\cc$
	consists in {\em exploring} one strongly connected component~$H$
of~$G_\cc$:
	The closed walks keep inside $H$, but may visit any node in~$H$.
For that, we first introduce for a strongly connected graph~$G$ the
	{\em exploration penalty\/} of~$G$,~$\EP(G)$, as the smallest integer~$e$
	such that for any node~$i$ and any integer~$n\geqslant e$ that is a 
	multiple of $G$'s cyclicity, there is a closed walk of length~$n$ 
	starting at~$i$.
The exploration penalty can be
	seen as the transient of diagonal entries in the sequence of Boolean matrix
	powers of the graph's adjacency matrix.
For us, it constitutes a threshold to pump walk
	lengths in multiples of the cyclicity.
We prove that $\EP(G)$ is finite, and from Brauer's Theorem~\cite{Bra42} we derive an 
		upper bound on $\EP(G)$ that is  quadratic in the number of nodes of~$G$.
This generalizes a theorem by Denardo~\cite{denardo} for 
		strongly connected graphs that are primitive, i.e., with cyclicity equal to 1.

\begin{theorem}\label{thm:EP}	
Let $G$ be a strongly connected graph with $N$ nodes, of girth~$g$ and
	cyclicity~$\gamma$.
The exploration penalty of $G$, denoted $\EP$, is finite and satisfies the inequality 
	$$\EP \leq \min \Big\{ N+(N-2)g\, , \, 2\frac{g}{\gamma}N -
			\frac{g}{\gamma}- 2g + \gamma\Big\}\enspace.$$
\end{theorem}

After proving Theorem~\ref{thm:EP}, the authors learned that the problem of
bounding the exploration penalty has already been  studied by
several authors (e.g., see \cite{LS93} for a survey).  Two bounds that do not
include the girth~$g$ as a parameter were given by Wielandt~\cite{Wie50} for
primitive graphs and by Schwarz~\cite{Sch70} for the general case.  Wielandt's
bound on the exploration penalty of a primitive strongly connected graph
with~$N$ nodes is called the {\em Wielandt number\/} $W(N)=N^2-2N+2$.  Schwarz
generalized this result to arbitrary cyclicities~$\gamma$ and arrived at a
bound of $\gamma\cdot W\big(\lfloor N/\gamma\rfloor\big) + (N \bmod \gamma)$.
To the best of our knowledge, our new bound in Theorem~\ref{thm:EP} is the
first one for non-primitive graphs that includes the girth~$g$ as a parameter.
In general, it is incomparable with the bound of Schwarz and shows the effect
of the girth~$g$ on the exploration penalty as the leading term in Schwarz'
bound is $N^2/\gamma$ whereas ours is at most $2Ng/\gamma$.

\medskip

The rest of this section is devoted to the proof of Theorem~\ref{thm:EP}.
If~$\gamma=1$, then $N+(N-2)g \leqslant 2gN/\gamma - g/\gamma -2g + \gamma$,
	and the inequality~$\EP\leqslant N + (N-2)g$ is actually a result by
	Denardo~\cite[Corollary~1]{denardo}.
Otherwise~$\gamma\geqslant 2$, and we easily check that 
	$N+(N-2)g \geqslant 2gN/\gamma - g/\gamma -2g + \gamma$.
In this case, we thus have to prove the inequality~$\EP\leqslant 2gN/\gamma - g/\gamma -2g + \gamma$.

For any pair of nodes $i$ and $j$, let ${\mathbf N}_{i,j}$ be the set of integers defined by
	$${\mathbf N}_{i,j} = \{ n\in \mathds{N}^* \mid \Pa^n(i,j) \neq\emptyset\} \enspace.$$
Clearly each ${\mathbf N}_{i,i}$ is nonempty and closed under addition; 
	let $d_i = \gcd ({\mathbf N}_{i,i})$.
Since $G$ is strongly connected, 
	$$	\gamma = \gcd (\{ d_i \mid i \text{ is a node in }G \}) \enspace.$$

Let $\mathbf{N}$ be any nonempty set of positive integers.
We call a subset $\mathbf{A} \subseteq \mathbf{N}$  a {\em gcd-generator\/} of 
	$\mathbf{N}$ if $\gcd(\mathbf{A}) = \gcd(\mathbf{N})$.\footnote{%
As $\mathds{Z}$ is Noetherian, any nonempty set of positive integers admits a
	finite gcd-generator.}

\begin{lemma}\label{lem:semigroup}
A nonempty set\/ $\mathbf N$ of positive integers that is closed under addition contains all but a finite
		number of multiples of its greatest common divisor.
Moreover, if\/ $\{a_1,\dots, a_k \}$ is a finite gcd-generator of\/~$\mathbf N$ with $a_1 \leq \dots \leq a_k$,
	then any multiple $n$ of $d=\gcd({\mathbf N})$ such that $n\geq (a_1 -d)(a_k -d)/d$
	is in\/ $\mathbf N$.
\end{lemma}
\begin{proof}
Consider the set $\mathbf M$ of all the elements in $\mathbf N$, divided by $d =\gcd ({\mathbf N})$.
By Brauer's Theorem~\cite{Bra42}, we know that every integer $m \geq (\frac{a_1}{d}-1)(\frac{a_k}{d}-1)$
	is of the form $$ m = \sum_{i=1}^k x_i  \frac{a_i}{d}$$
	where each $x_i$ is a nonnegative integer.
Since $\mathbf N$ is closed under addition, it follows that  every  multiple of $d$ that is greater or equal 
	to $(a_1-d)(a_k-d)/d$ is in $\mathbf N$.
In particular, 	all but a finite number of multiples of $d$ are in $\mathbf N$.
\end{proof}

\begin{lemma}\label{lem:Ni,j}
For any node $i$, $d_i = \gamma$.
Moreover, for any pair of nodes $i,j$, all the elements in ${\mathbf N}_{i,j}$ have the same
	residue modulo $\gamma$.
\end{lemma}
\begin{proof}
Let $i,j$ be any pair of nodes, and let $a \in {\mathbf N}_{i,j}$ and $b \in {\mathbf N}_{j,i}$.
The concatenation of a walk from~$i$ to~$j$ with a walk from~$j$ to~$i$ 
	is a closed walk starting at~$i$.
Hence $a + b \in {\mathbf N}_{i,i}$. 
From Lemma~\ref{lem:semigroup}, we know that ${\mathbf N}_{j,j}$ contains all the multiples of $d_j$ 
	greater than some integer.
Consider any such multiple $k d_j$ with $k$ and $d_i$ relatively prime integers.
By inserting one corresponding closed walk at node $j$ into the closed walk at $i$
	with length $a+b$, we obtain a new closed walk starting at $i$, i.e.,
	$a + k d_j + b \in {\mathbf N}_{i,i}$.
It follows that $d_i$ divides both $a+b$ and $a + k d_j + b$, and so~$d_i$ divides~$d_j$.
Similarly, we prove that~$d_j$ divides~$d_i$, and so $d_i = d_j$.
Because~$\gamma$ is the gcd of the~$d_i$'s, the common value of the $d_i$'s  is actually equal to
$\gamma$.

Let $a$ and $a'$ be two  integers in ${\mathbf N}_{i,j}$.
The above argument shows that both $a + b$ and $a' +b$ are in ${\mathbf N}_{i,i}$.
Hence $\gamma$ divides $a + b $ and $a'+ b$, and so also $a-a'$.
\end{proof}

\begin{lemma}\label{lem:Ni,i:generator}
For any node $i$, the set\/  ${\mathbf N}_{i,i}$ admits a gcd-generator 
	that contains the lengths of all the cycles starting at $i$, 
	and whose all elements $n$ satisfy the  inequality 
	$ g \leq n \leq 2N -1$.
\end{lemma}
\begin{proof}
Let $i$ be any node of $G$, and let $C_0$ be any cycle.
Let $W_1$ be one of the shortest paths from $i$ to $C_0$,  and set $j= \End(W_1)$.
Without loss of generality, $\Start(C_0)=j$.
By definition, $\ell(W_1) \leq N-\ell(C_0)$.
Then consider a path $W_2$ from $j$ to $i$, and the two closed walks
	$$W = W_1 \cdot W_2 \mbox{ and } W' = W_1 \cdot C_0 \cdot W_2\enspace.$$ 
Note that $$\ell(W) \leq \ell(W') \leq 2N-1 \enspace.$$
Moreover if the walk $W$ is nonempty, then 
	$$\ell(W)\geqslant g \enspace,$$
	because~$W$ is closed.
In the particular case $i$ is a node of $C_0$,  $W$ is the empty walk starting
	at $i$, $W'$ reduces to $C_0$,
	and $\ell(W')$ is the length of the cycle $C_0$. 

Let $\mathbf{N}_i$ be the set of the lengths of the nonempty closed walks~$W$ and~$W'$ when 
	considering all the cycles~$C_0$ in~$G$.
Then, $\mathbf{N}_i$ contains the length of all the cycles starting at $i$.
Let $\delta_i =\gcd(\mathbf{N}_i)$.
Since $\mathbf{N}_i \subseteq \mathbf{N}_{i,i}$,  $d_i$ divides $\delta_i$.
Conversely, let $C_0$ be any cycle, and let $W$ and $W'$
	be the two closed walks starting at node $i$ defined above;~$\delta_i$ divides both
	$\ell(W)$ and $\ell(W')$, and so divides $\ell(W') - \ell(W) = \ell(C_0)$.
Hence,~$\delta_i$ divides the length of any cycle, i.e., $\delta_i$ divides $\gamma$.
By Lemma~\ref{lem:Ni,j}, it follows that~$\delta_i$ divides~$d_i$.
Consequently, $\delta_i = d_i$, i.e., $\mathbf{N}_i$ is a gcd-generator of $\mathbf{N}_{i,i}$.
\end{proof}

\begin{lemma}\label{lem:bernadette:denardo}
For any node $i$ and any integer 
	$n$ such that $n$ is a multiple of $\gamma$  and 
	$n \geqslant 2Ng/\gamma- g/\gamma - 2g + \gamma$,
	there exists a closed walk of length~$n$ starting at~$i$. 
\end{lemma}
\begin{proof}
Let $i$ be any node, and let $C_0$ be any cycle
	such that $\ell(C_0) = g$.
Let~$W_1$ be one of the shortest walks from $i$ to $C_0$,  and set $j= \End(W_1)$.
Without loss of generality, $\Start(C_0)=j$.
By definition, $\ell(W_1) \leq N- g$.
Then consider a path~$W_2$ from $j$ to $i$; we have $\ell(W_2) \leq N-1$.
The walk $W_1 \cdot W_2$ is closed at node $i$, and so $\gamma$ divides 
	$\ell(W_1) + \ell(W_2)$.
Hence, if~$\gamma$ divides some integer $n$, then $\gamma$ also divides $n - \ell(W_1) - \ell(W_2)$.
It is $g\in \mathbf{N}_{j,j}$.
By Lemma~\ref{lem:Ni,i:generator}, there exists a gcd-generator~$\mathbf{N}_j$ of~$\mathbf{N}_{j,j}$ such that $g\in\mathbf{N}_j$ and $g\leqslant n\leqslant 2N-1$ for all $n\in\mathbf{N}_j$.

By Lemma~\ref{lem:semigroup}, for any~$n$ such that
	$n' = n - \ell(W_1) - \ell(W_2)$ is a multiple of~$\gamma$ and
	$$ n' \geq \gamma\,\left(\frac{g}{\gamma} -1\right)
\left(\frac{2N-1}{\gamma} -1\right)\enspace,$$
	there exists a closed walk $C$ starting at node $j$ of length $\ell(C) = n'$.
Note that
\[
\gamma\left(\frac{g}{\gamma}-1\right)\left(\frac{2N-1}{\gamma}-1\right)+(N-g)+(N-1)
= 2\frac{g}{\gamma}N - \frac{g}{\gamma} - 2g + \gamma \enspace.\]
In this way, for any integer $n \geqslant 2Ng/\gamma - g/\gamma -2g+ \gamma$ 
 	that is a multiple of $\gamma$, we construct 
	$W = W_1 \cdot C \cdot W_2$ that is a closed walk at node $i$ of length $n$.
\end{proof}

Theorem~\ref{thm:EP} immediately follows from Lemma~\ref{lem:bernadette:denardo}.

\section{Repetitive and Explorative Transience Bounds}\label{sec:bounds}

We now follow the strategy laid out in
     Section~\ref{sec:bounds:outline}  to prove two new bounds on
     system transients.
They mainly differ in step~4 of the strategy, namely, in the way one
     completes the reduced walk~$\Red_{d,k}(W)$ with critical closed
     walks to reach the desired length~$n$.
Naturally this has implications on the appropriate choices for the
     walk reduction parameters~$d$ and~$k$ used in step~3.
	
Let~$A$ be an irreducible $N\times N$ max-plus matrix with $\lambda(A) = 0$, 
	and let $v$ be a vector in $\IR^N$.
Recall that~$\pi$ is chosen to be the least common multiple of cycle lengths in 
	the critical subgraph~$G_\cc$.
Let~$i$ be any node, and let~$B$ and $n$ be two positive integers
	such that $n \geqslant B \geqslant B_\cc$.
Since $\lambda(A) = 0$, there exists a walk $W$ that is an $\realrem{B}{n,\pi}$-realizer 
	for~node $i$.
By definition of~$\realrem{B}{n,\pi}$, $\ell(W) \geqslant B$, and
	walk~$W$ is a $\{\ell(W)\}$-realizer for~node~$i$. 
Proposition~\ref{prop:n:zero} shows that $W$ contains at least one critical node~$k$.
Let~$H$ denote the strongly connected component of~$G_\cc$ containing~$k$.

We consider~$d$ to be any divisor of $\pi$.
By construction,  $\hat{W}= \Red_{d,k}(W)$ is obtained  by removing a collection 
	of cycles from~$W$, and starts at the same node~$i$ as~$W$.
Since $\lambda(A)=0$, this implies 
	\begin{equation}\label{eq:weight}
		A_v (\hat{W}) \geqslant A_v(W) \enspace.
	\end{equation}
Moreover, walk~$\hat{W}$ contains the critical node~$k$, and
	its length $\ell(\hat{W})$ is in the same residue class modulo~$d$ as~$\ell(W)$.
By Theorem~\ref{thm:red:upper:bound}, we have
\begin{equation}\label{eq:length}
	\ell(\hat{W}) \leqslant (d-1) + 2d \cdot (N-1) \enspace.
	\end{equation}
	
For the {\em repetitive\/} bound, we use a single critical cycle~$C$ to
	complete~$\hat{W}$; see Figure~\ref{fig:lem:realizer:non1}.
Let~$C$ be a cycle with length equal to the girth~$g(H)$.
We can assume that~$k$ is a node of~$C$: In case~$k$ is not a node of $C$,
        we modify $W$ by inserting~$\pi$ copies of a critical closed walk in $H$
        that connects~$W$ to~$C$.
Indeed, the addition of this critical closed walk changes neither the residue class 
	modulo~$\pi$ nor the $A_v$-weight since $\lambda(A) =0$.
We now choose
\[d=g(H) \enspace.\]
From~(\ref{eq:length}), we derive that  
	$n \geqslant \ell(\hat{W})$ when  $B \geqslant (g(H)-1) + 2g(H) \cdot (N-1)$,
	and we complete the reduced walk~$\hat{W}$ to
	length~$n$ by adding copies of~$C$.

For the {\em explorative\/} bound, we choose
\[d=\gamma(H)\]
 and use the
	definition of the exploration penalty~$\EP(H)$.
From~(\ref{eq:length}), we derive that $n \geq \ell(\hat{W})+ \EP(H)$ when
	$B \geqslant (\gamma(H)-1) + 2\gamma(H) \cdot (N-1) + \EP(H)$.
By definition of $\EP(H)$ and since $n-\ell(\hat{W}) \ge \EP(H)$, we can 
        complete~$\hat{W}$ to
	length~$n$ by a critical closed walk in~$H$; see Figure~\ref{fig:lem:realizer}.
	
In each of the two completions, the resulting walk is of length~$n$, starts at node~$i$,  
	and ends at the same node as $\hat{W}$.
With~(\ref{eq:weight}), we deduce that its $A_v$-weight is at least $A_v(W)$.
Thereby, it is an $\realrem{B}{n,\pi}$-realizer for node~$i$ of length~$n$.
By Proposition~\ref{prop:final:step}, the repetitive and explorative completions 
	finally give the following upper bounds on system transients.

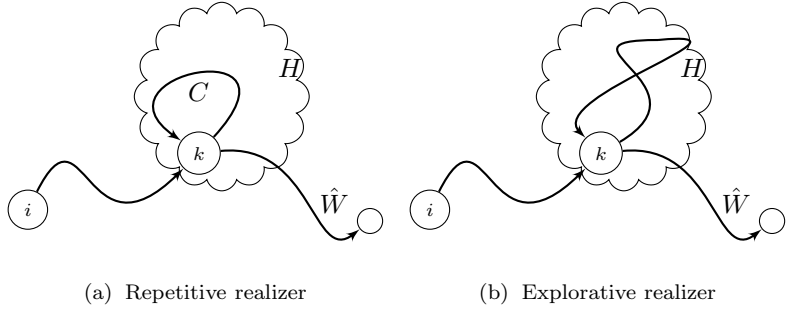
\begin{figure}[bht]
\centering
\subfigure[$\!$ Repetitive realizer]{
\begin{tikzpicture}[>=latex',scale=0.75]
        \node[shape=circle,draw] (i) at (-3,-1) {$\scriptstyle i$};
        \node[shape=circle,draw] (k) at (0,0) {$\scriptstyle k$};
        \node[shape=circle,draw] (j) at (3,-1.2) {};
        \draw[thick,->] (i) .. controls (-2,1) and (-2,-2)
..node[midway,right]{}  (k);
        \draw[thick,->] (k) .. controls (2,0.2) and (2,-2)  ..
node[midway,below,right=1pt]{$\hat{W}$}  (j);
        \draw[thick,->] (k) .. controls +(2,2.2) and +(-2,1.4)  .. node[below]
{$C$}  (k);
        \node[cloud, cloud puffs=16, draw,minimum width=2.3cm, minimum
height=2.5cm] at (0.4,1) {};
        \node at (1.6,1.5) {$H$};
\end{tikzpicture}
\label{fig:lem:realizer:non1}
}
\subfigure[$\!$ Explorative realizer]{
\begin{tikzpicture}[>=latex',scale=0.75]
        \node[shape=circle,draw] (i) at (-3,-1) {$\scriptstyle i$};
        \node[shape=circle,draw] (k) at (0,0) {$\scriptstyle k$};
        \node[shape=circle,draw] (j) at (3,-1.2) {};
        \draw[thick,->] (i) .. controls (-2,1) and (-2,-2)  ..
node[midway,above]{}  (k);
        \draw[thick,->] (k) .. controls (2,0.2) and (2,-2)  ..
node[midway,below,right=1pt]{$\hat{W}$}  (j);
        \draw[thick] (k) .. controls (2,1.2) and (-1,2)  .. node[near
start,left]{}  (1,2);
        \draw[thick,->] (1,2) .. controls (3,2.2) and (-1,1)  ..  (k);
        \node[cloud, cloud puffs=16, draw,minimum width=2.3cm, minimum
height=2.5cm] at (0.4,1) {};
        \node at (1.6,1.5) {$H$};
\end{tikzpicture}
\label{fig:lem:realizer}
}
\caption{Repetitive and explorative realizers}
\end{figure}

\begin{theorem}[Repetitive Bound]\label{thm:nonexplorative}
Denoting by~$\hat{g}$ the maximum girth of strongly connected components
	of~$G_\cc$, the transient of the linear max-plus system $\langle A,v\rangle$ is at most
\[ \max \left\{ \frac{\lVert v \rVert + \big(\Delta_\nc - \delta\big)
\cdot (N-1) }{ \lambda - \lambda_\nc}  \ ,\  (\hat{g}-1) + 2\,\hat{g}\cdot(N-1) \right\} \enspace.\]
\end{theorem}

\begin{theorem}[Explorative Bound]\label{thm:explorative:nonprimitive}
Denoting by~$\hat{\gamma}$ and~$\hat{\EP}$ the maximum cyclicity and
maximum exploration penalty of strongly connected components
of~$G_\cc$, respectively, the transient of the  linear max-plus system $\langle A,v\rangle$ is at
most
\[ \max \left\{ \frac{\lVert v \rVert + \big(\Delta_\nc - \delta\big)
\cdot (N-1) }{ \lambda - \lambda_\nc}  \ ,\  (\hat{\gamma}-1) +
2\,\hat{\gamma}\cdot(N-1) + \hat{\EP} \right\} \enspace.\]
\end{theorem}

Because~$\hat{g}$ is greater or equal to~$\hat{\gamma}$, the two
     bounds represent a tradeoff between choosing a larger
     multiplicative term versus the addition of the term~$\hat{\EP}$.
It depends on the critical subgraph~$G_\cc$ which of the two bounds is
     better, and our two bounds are  thus incomparable in general: As
     an example for which the explorative bound is lower than the
     repetitive bound, consider the family of graphs $E_k$ depicted in
     Figure~\ref{fig:E}: $E_k$ consists of two joint cycles of length
     $k$ and $k+1$, respectively.
All edges have zero weight.
Independent of the initial vector $v$, the critical bound is $N$,
     since $\lambda_{\nc} = -\infty$.
With $N=2k$, $\hat{g}=k$, and $\hat{\gamma}=1$, the repetitive bound
     is $4k^2-k-1$, and the explorative bound is at most $2k^2+4k-2$.
For $k\ge 3$ the explorative bound is strictly lower than the
     repetitive bound.
Conversely, the repetitive bound is lower than the explorative bound,
     if we add a self-loop at the node that is shared by the
     two cycles in the above example.

\begin{figure}[bt]
\centering
\begin{tikzpicture}[>=latex',scale=1]

\def \k {6}
\def \kp {7}
\def \radius {1cm}
\def \margin {15} %

\foreach \s in {1,...,\k}
{
  \node[draw, circle] at ({360/\k * (\s - 1)}:\radius) {};
  \draw[->, >=latex] ({360/\k * (\s - 1)+\margin}:\radius) 
      arc ({360/\k * (\s - 1)+\margin}:{360/\k * (\s)-\margin}:\radius);
}

\begin{scope}[shift={(-2cm,0)}]
\foreach \s in {1,...,\kp}
{
  \node[draw, circle] at ({360/\kp * (\s - 1)}:\radius) {};
  \draw[->, >=latex] ({360/\kp * (\s - 1)+\margin}:\radius) 
    arc ({360/\kp * (\s - 1)+\margin}:{360/\kp *
      (\s)-\margin}:\radius);

}
\end{scope}

\node at (0,0) {$k$};
\node at (-2cm,0) {$k+1$};

\end{tikzpicture}
\caption{Graphs $E_k$}
\label{fig:E}
\end{figure}
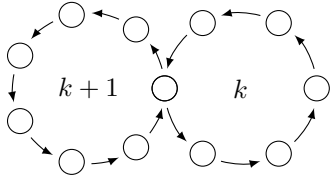

Interestingly, the two terms in our transience bounds
	that are due to the repetitive and explorative
	completions are both at most quadratic: this is obvious for 
	the repetitive term, %
	and is an immediate corollary of Theorem~\ref{thm:EP} for the explorative term.
In the case of integer matrices, for a given initial vector, both the repetitive and the
	explorative bounds are in $O(\lVert A \rVert \cdot N^3)$ since the critical bound
	itself is in $O(\lVert A \rVert \cdot N^3)$ in this case
        (see Equation~(\ref{eq:int_bound})).

Hartmann and Arguelles~\cite{hartmann:arguelles} established the best
previously known bound on system transients.
Their approach includes passing to the max-balancing~\cite{SS91} of~$G$ and
considering an increasing sequence of threshold graphs which all include the
critical subgraph. 
Their technique to increase the length of maximum weight walks is comparable
to our repetitive pumping technique.
They proved that the
 transient of system $\langle A,v\rangle$ is upper-bounded by $\max\big( (\lVert v\rVert + \lVert
A\rVert \cdot N ) / (\lambda - \lambda^0) \,,\,2N^2 \big)$
where~$\lambda^0$
 is defined in terms of the max-balancing of~$G$.
The first term in their bound is
in general incomparable with our critical bound,
whereas the second term, namely $2N^2$, is always strictly
larger than the second term in each of our two bounds and makes their bound at
least quadratic in~$N$.
Trivially, the minimum of our two bounds, and of  Hartmann and
Arguelles' bound, yields the best currently known bound.

\section{Matrix vs.\ System Transients}\label{sec:matrix}

As explained in Section~\ref{sec:bounds:outline},  we can follow the
     same strategy as for system transients to bound matrix
     transients.
For an $N \times N$ max-plus matrix~$A$, this leads to an upper bound that is 
     in $O\big(\lVert A\rVert \cdot N^2/(\lambda-\lambda_\nc)$, but
     gives no hint on the relationships between the transient of
     max-plus matrix~$A$, and the transients of the max-plus
     systems~$\langle A,v\rangle$.

In this section, we show that  the  transient of matrix~$A$ is
     actually equal to the transient of a specific system~$\langle
     A,v\rangle$ where  $\lVert v \rVert $ is in
     $O\big(\lVert A\rVert\cdot N^2\big)$, provided the system
     transient is sufficiently large, namely at most equal to some
     term quadratic in~$N$.
Combined with our upper bounds on the system transient established in
     Theorems~\ref{thm:nonexplorative}
     and~\ref{thm:explorative:nonprimitive}, this gives two upper bounds
     on the matrix transient which are also in $O\big(\lVert A\rVert
     \cdot N^2/(\lambda-\lambda_\nc)\big)$, and so in
     $O(\lVert A\rVert \cdot N^4)$ for integer matrices.

Let $n_A$ and $n_{A,v}$ denote the transient of matrix~$A$ and the  transient of 
	system~$\langle A,v\rangle$, respectively.
Obviously, $n_A$ is an upper bound on the $n_{A,v}$'s.
Conversely, the equalities 
	$A_{i,j}^{\otimes n} = \big( A^{\otimes n} \otimes e^{j} \big)_i$,
	where the $e^j$'s are the unit vectors defined by
	$e^{j}_i = 0$ if~$i=j$ and~$e^{j}_i=-\infty$ otherwise, show that
	$\max \big\{ n_{A,e^{j}} | j \in \{1,\cdots,N\} \big\} \geqslant n_A$.
	Hence, 
		\begin{equation}\label{eq:nA}
			\sup\big\{ n_{A,v} | v \in \IRmax^N \big\} = n_A \enspace.\notag
			\end{equation}
We now seek a similar expression of $n_A$, but with finite initial vectors~$v$, 
	i.e., with $v\in \IR^N$.
Reusing the notation~$\hat{\gamma}$ and~$\hat{\EP}$ from
	Theorem~\ref{thm:explorative:nonprimitive}, we define:
\begin{align}
 \Bunu &= 2(N - 1) + \hat{\EP} + (\EP(G) + \hat{\gamma}-1),  \notag\\
 \mu &= \sup \left\{ A_{i,h}^{\otimes n} - A_{i,j}^{\otimes n} \mid h,i,j \mbox{ nodes of } G\ ,\ 
  n\geqslant \Bunu\ ,\ A_{i,j}^{\otimes n} \neq -\infty \right\} \notag
\end{align}
Clearly $\mu$ is finite, i.e., $\mu\in \IR$.
Then we consider the {\em $\mu$-truncated\/} unit vectors obtained by replacing the infinite 
	entries of the $e^j$'s by~$-\mu$.
	
In Proposition~\ref{prop:mu} below, we show that if~$B\geqslant \Bunu$ and~$B$
	is a bound on the system transients for all $\mu$-truncated unit vectors,
	then~$B$ is also a bound on the matrix transient.  
A technical difficulty in the proof lies in the fact that, contrary to the sets $\Pa^n(\ito)$
	which occur in the expression of the $i$-th  component of linear systems,  the sets $\Pa^n(i,j)$
	that we consider for matrix powers may be empty. 
The next two lemmas deal with this technicality.

\begin{lemma}\label{lem:walkpenality}%
For any pair of nodes $i,j$ of~$G$ and any integer  $n \geqslant \EP(G) +
\gamma(G) + N-2$,
	there exists a walk~$W$ from~$i$ to~$j$ such that $n-\ell(W) \in \{0,
\dots, \gamma(G)-1\}$.
\end{lemma}
\begin{proof}
Let~$i, j$ be any two nodes, and let $W_0$ be a path from $i$ to $j$.
For any integer $n$, consider the residue $r$ of $n-\ell (W_0)$ modulo $\gamma(G)$.
By definition of $\EP(G)$, if $n-\ell (W_0) - r \geqslant \EP(G)$, then there
	exists a closed walk $C$ starting at node $j$ with length equal to 
	$n-\ell (W_0) - r$.
Then, $W_0 \cdot C$ is a walk from $i$ to $j$ with length $n-r$, where
	$r \in \{0, \dots, \gamma(G)-1\}$.
The lemma follows since $n-\ell (W_0) - r \geqslant \EP(G)$ as soon as 
	$n \geqslant \EP(G) + (N-1) + \gamma(G) - 1$.
\end{proof}

\begin{lemma}\label{lem:ep:g}
Let $n$ be any integer such that $n\geqslant \EP(G)+ \gamma(G) + N -2$.
Then $A_{i,j}^{\otimes (n+\gamma(G))}=-\infty$ if and only if $A_{i,j}^{\otimes n} = -\infty$.
\end{lemma}
\begin{proof}
It is equivalent to claim that $\Pa^{n+\gamma(G)}(i,j)=\emptyset$ if and
     only if $\Pa^n(i,j)=\emptyset$ for any integer~$n\geqslant \EP(G)+
\gamma(G) + N -2$.

Suppose $\Pa^{n+\gamma(G)}(i,j) \neq \emptyset$, and let  $W_0\in\Pa^{n+\gamma(G)}(i,j)$. 
By Lemma~\ref{lem:walkpenality}, there exists a walk $W\in\Pa(i,j)$ such that 
	$n=\ell(W)+r$ with $r\in\{0,1,\dots,\gamma(G)-1\}$.
Lemma~\ref{lem:Ni,j} implies that~$\gamma(G)$ divides $\ell(W_0)-\ell(W) =
(n+\gamma(G)) - (n-r) = \gamma(G)+r$; hence~$\gamma(G)$ divides~$r$.
Therefore,~$r=0$, i.e., $\ell(W) =n$ and thus~$\Pa^n(i,j)\neq\emptyset$.

The converse implication is proved similarly.
\end{proof}

\begin{proposition}\label{prop:mu}
If~$n\geqslant \Bunu$ and $A^{\otimes (n+\gamma)}\otimes v=A^{\otimes n}\otimes v$
for all $\mu$-truncated unit vectors~$v$, then $A^{\otimes (n+\gamma)}=A^{\otimes n}$.
\end{proposition}
\begin{proof}
Let~$i$ and~$j$ be nodes in $G$, and let~$n$ be an integer such that $n \geqslant \Bunu$.
Further let~$v$ be the $\mu$-truncated unit vector with~$v_j=0$ and~$v_h=-\mu$
for~$h\neq j$.
Since~$\Bunu \geqslant \EP(G) + \gamma(G) + N-2$ and~$\gamma=\gamma(G_\cc)$ is a
multiple of~$\gamma(G)$,
	we derive from Lemma~\ref{lem:ep:g} that $A_{i,j}^{\otimes n+\gamma}=-\infty$ if 
	and only if $A_{i,j}^{\otimes n}=-\infty$.
There are two cases to consider:
\begin{enumerate}
\item $A_{i,j}^{\otimes n}= -\infty$ and $A_{i,j}^{\otimes n+\gamma}= -\infty$.
In this case, 	$A_{i,j}^{\otimes n+\gamma} = A_{i,j}^{\otimes n}$ trivially holds.
\item  $A_{i,j}^{\otimes n}\neq-\infty$ and $A_{i,j}^{\otimes n+\gamma}\neq -\infty$.
Recall that 
	$$\left( A^{\otimes n} \otimes v \right)_i = \max
\big\{A_{i,h}^{\otimes n} + v_h \ | \ h\in \{1,\cdots,N\}\big\} \enspace.$$
By definition of $\mu$ and $v$,  for any node $h\neq j$,
	\[ A_{i,h}^{\otimes n} - A_{i,j}^{\otimes n} \leqslant \mu =
v_j - v_h \enspace. \]
It follows that 
     \[ \left( A^{\otimes n} \otimes v \right)_i = A_{i,j}^{\otimes n} + v_j \enspace.\]
As  $n +\gamma \geqslant n$, we similarly have
\[ A_{i,j}^{\otimes n +\gamma} = \big( A^{\otimes n+\gamma} \otimes v \big)_i - v_j
= \left( A^{\otimes n} \otimes v \right)_i - v_j = A_{i,j}^{\otimes n} \enspace. \]
Thus  $A_{i,j}^{\otimes n+\gamma} =A_{i,j}^{\otimes n}$  holds also in this case.
\end{enumerate}
\end{proof}

The key point for establishing our bound on matrix transients is the following
upper bound on~$\mu$, which is quadratic in~$N$. The proof uses the pumping
technique developed for the explorative bound twice.

\begin{proposition}\label{prop:mu:upper:bound}
$\displaystyle \mu \leqslant \lVert A\rVert \cdot \Bunu$
\end{proposition}
\begin{proof}
First, we observe that each term in the inequality to show is invariant under
	substituting $A$ by $\overline{A}$.
Hence we assume that $\lambda=0$. 
It follows that
\begin{equation}\label{eq:mu:upper:bound:upper}
A_{i,h}^{\otimes n} \leqslant \Delta\cdot (N-1) \leqslant \Delta \cdot \Bunu\enspace.
\end{equation}

We now give a lower bound on~$A_{i,j}^{\otimes n}$ in the case
     that it is finite, i.e., if~$\Pa^n(i,j)\neq\emptyset$.
Let~$k$ be a critical node in the strongly connected component~$H$ of~$G_\cc$ with minimal
     distance from~$i$ and let~$W_1$ be a shortest path from~$i$
     to~$k$.
Further, let~$W_2$ be a shortest path from~$k$ to~$j$.
Let $r$ denote the residue of $n - \ell(W_1\cdot W_2) - \EP(G)$ modulo $\gamma(H)$,
	and let $t= n - \ell(W_1\cdot W_2) - \EP(G) -r$.
Since~$t\equiv0\pmod{\gamma(H)}$, and 
	$$t \geqslant \Bunu -  2(N-1) - \EP(G) - \big(\gamma(H) - 1\big) \geqslant
\hat{\EP} \geqslant  \EP(H) \enspace,$$
	there exists a closed walk~$C_\cc$ of length~$t$ in
	component~$H$ starting at node~$k$.
Let $s= \EP(G) + r$; then, $s \geqslant \EP(G)$.
Moreover, $s= n - \ell(W_1\cdot C_\cc\cdot W_2)$, and 
	$W_1\cdot C_\cc\cdot W_2 \in \Pa(i,j)$.
By Lemma~\ref{lem:Ni,j}, it follows that $\gamma(G)$ divides $s$,
because~$\Pa^n(i,j)\neq\emptyset$.
Hence there exists a closed walk~$C_{\nc}$ of length~$s$ starting
     at node~$j$.

Now define~$W=W_1\cdot C_\cc\cdot W_2\cdot C_{\nc}$.
Clearly,~$\ell(W) =n$ and
	$$ n(W) \geqslant \delta \cdot( n - t
     ) \geqslant \delta\cdot \big( 2 (N-1) + \EP(G) + \gamma(H) -1
     \big) \enspace,$$
and so
	\begin{equation}\label{eq:mu:upper:bound:lower}
		A_{i,j}^{\otimes n} \geqslant  \delta\cdot \big( 2 (N-1) + \EP(G) + 
		\hat{\gamma} -1\big) \geqslant \delta\cdot \Bunu \enspace.
	\end{equation}
From~\eqref{eq:mu:upper:bound:upper} and~\eqref{eq:mu:upper:bound:lower} 
	follows $\mu \leqslant (\Delta - \delta)\cdot \Bunu = \lVert A\rVert \cdot \Bunu$. 
\end{proof}

\begin{figure}[tb]
\centering
\begin{tikzpicture}[>=latex',scale=0.8]
	\node[shape=circle,draw] (i) at (-4,-1) {$\scriptstyle i$};
	\node[shape=circle,draw] (k) at (0,0) {$\scriptstyle k$};
	\node[shape=circle,draw] (j) at (4,-1.2) {$\scriptstyle j$};
	\draw[thick,->] (i) .. controls (-3,1) and (-2,-2)  .. node[midway,above]{${W}_1$}  (k);
	\draw[thick,->] (k) .. controls (2,0.2) and (3,-2)  .. node[midway,above]{${W}_2$}  (j);
	\draw[thick] (k) .. controls (2,1.2) and (-1,2)  .. node[near
start,left]{$C_{\cc}$}  (1,2);
	\draw[thick,->] (1,2) .. controls (3,2.2) and (-1,1)  ..  (k);
	\node[cloud, cloud puffs=16, draw,minimum width=3.1cm, minimum height=3.5cm] at (0.6,1) {};
	\node at (1.8,1.5) {$H$};
	\draw[thick] (j) .. controls (4,1.2) and (3,1)  .. node[near start,right]{$C_{\nc}$}  (4,1);
	\draw[thick,->] (4,1) .. controls (5,1.2) and (3,1)  ..  (j);
\end{tikzpicture}
\caption{Walk~$W$ in proof of Proposition~\ref{prop:mu:upper:bound}}
\label{fig:lem:mu:upper:bound}
\end{figure}
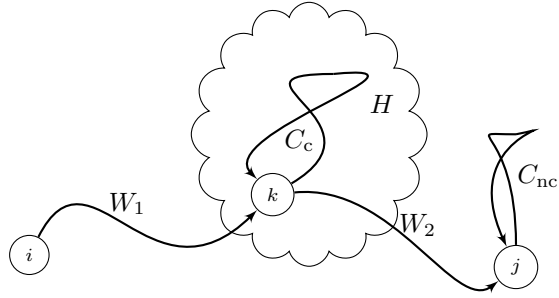

Combined with our upper bounds on the system transient established in 
	Theorems~\ref{thm:nonexplorative} and~\ref{thm:explorative:nonprimitive},
	Propositions~\ref{prop:mu} and~\ref{prop:mu:upper:bound} give a repetitive 
	upper bound and an explorative upper bound on the matrix transient.

\begin{theorem}\label{thm:matrix:system}
The transient of an irreducible matrix is at most equal to the minimum
	of the repetitive bound
	\[ \max \left\{\Bunu  \ ,\  \frac{\lVert A\rVert\cdot \Bunu + \big(\Delta_\nc - \delta\big)
	\cdot (N-1) }{ \lambda - \lambda_\nc}  \ ,\  (\hat{g}-1) + 2\,\hat{g}\cdot(N-1) 
	                                         \right\} \enspace,\]
	 and the explorative bound 
	\[ \max \left\{ \Bunu  \ ,\  \frac{\lVert A\rVert\cdot \Bunu + \big(\Delta_\nc - \delta\big)
	\cdot (N-1) }{ \lambda - \lambda_\nc}  \ ,\ (\hat{\gamma}-1) +
						2\,\hat{\gamma}\cdot(N-1) + \hat{\EP}\right\} \enspace,\]
	where $\Bunu = 2(N - 1) + \hat{\EP} + (\EP(G) + \hat{\gamma}-1)$.				
\end{theorem}

Note that by Theorem~\ref{thm:EP}, the term~$\Bunu$ in the above
     bounds is at most quadratic in~$N$.
Moreover it can be removed from the maximum when~$\lambda_\nc$ is
     finite, since in this case the critical bound dominates the
     term~$\Bunu$ as $\lambda-\lambda_\nc \leq \Delta - \delta = \lVert A\rVert$.

Further, from Theorem~\ref{thm:EP} we immediately obtain that the
     transient of an irreducible matrix is in 
	 $O\big(\lVert A\rVert \cdot N^2/(\lambda-\lambda_\nc)\big)$
	 if $\lambda_{\nc}$ is finite, and in~$O(N^2)$, otherwise.
In particular, for integer matrices the matrix transient is in
     $O(\lVert A\rVert \cdot N^4)$ for integer matrices.

\section{Applications}\label{sec:discussion}%

In this section we demonstrate how our transience bounds enable the
     performance analysis of various distributed systems, thereby
     obtaining simple proofs both of known and new results.

In Section~\ref{subsec:cyclic}, we discuss properties of optimal
     cyclic schedules of a set of tasks subject to a set of
     restrictions.
This problem arises, e.g., in manufacturing, time-sharing of
     processors in embedded systems, and design of compilers for
     scheduling loop operations for parallel and pipelined
     architectures.
By applying our transience bounds to a naturally arising special case
     of restrictions (with binary heights), we are able to state
     explicit upper bounds, and thereby asymptotic upper bounds, on
     the number of task executions from where on the schedule becomes
     periodic.

In Section~\ref{subsec:sync}, we discuss the transient behavior of the
$\alpha$ network synchronizer~\cite{Aw85}.
The $\alpha$-synchronizer constructs virtually synchronous rounds in a
strongly connected network of processes that communicate by message passing
      with constant transmission delays.
Its time behavior can be described by a max-plus linear system.
It has hence a periodic behavior and by
applying our results, we obtain upper bounds
     on the time from which on the system is periodic.
We show that our bounds are strictly better than those by Even and
     Rajsbaum~\cite{even:rajs}.
In the case of integer matrices considered by Even and Rajsbaum, our bounds are in $O(\lVert A\rVert \cdot N^3)$ 
which we show to be asymptotically tight.

In Section~\ref{subsec:fr}, we further exemplify the applicability of
     our results to distributed algorithms by deriving upper bounds on
     the termination time of the Full Reversal algorithm when used for
     routing~\cite{GB87}, and the time from which on it is periodic
     when used for scheduling~\cite{BG89}.

\subsection{Cyclic scheduling}\label{subsec:cyclic}

Cohen et al.~\cite{CMQV89} have observed that, in cyclic scheduling, the class of
{\em earliest schedules\/}
can be described as max-plus linear systems.
In this section, we show how to use this fact and our general bounds to derive explicit upper
bounds on transients of earliest schedules.

If a finite set~$\mathcal{T}$ of tasks (each of which
     calculates a certain function) is to be scheduled
repeatedly on different
	processes, precedence restrictions are implied by the data flow.
These restrictions are of the form that task~$i$  may start its
     number~$n$ execution only after task~$j$ has finished its
     number~$n-h$ execution.
A {\em schedule\/}~$t$ maps a pair $(i,n)\in\mathcal{T}\times\IN$ to a
nonnegative integer $t(i,n)$, the time the
     number~$n$ execution of task~$i$ is started.
Formally, if~$P_i$ denotes the processing time of task~$i$, then a {\em
restriction}~$R$
between two tasks~$i$ and~$j$ is an inequality of the form
\begin{equation} 
\label{eq:depend}
\forall n\geqslant h_{R}:\,\, t(i,n) \ge t(j,n-h_{R}) + P_j
\end{equation}%
where $h_{R}$ is called the {\em height\/}
of restriction~$R$ and $P_j$ is its {\em weight}.

A {\em uniform graph\/}~\cite{HM95} describes a set of tasks and restrictions.
Formally, it is a quadruple
     $G^u=(\mathcal{T},E,p,h)$ such that $(\mathcal{T},E)$ is a
     directed (multi-)graph, and $p:E\to \IN^*$ and $h:E\to\IN$ are two
     functions, the {\em weight\/} and {\em height\/} function,
     respectively.
For a walk~$W$ in~$G^u$, let $p(W)$ be the sum of the weights of its
     edges and $h(W)$ the sum of the heights of its edges.
An edge from~$i$ to~$j$ corresponds to a restriction~$R$ between~$i$
and~$j$ of the form~\eqref{eq:depend}.
All incoming edges of a node~$j$ in $\mathcal{T}$ have the
     same weight, namely~$P_j$.
An example of a uniform graph is Figure~\ref{fig:dep_graph}.
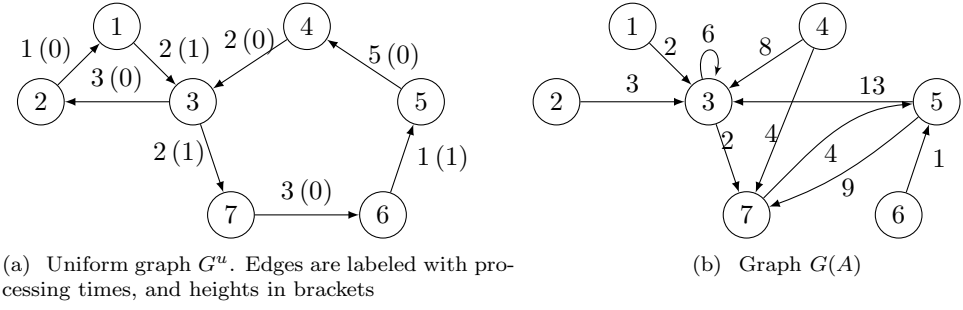
\begin{figure}[hbt]
\centering
\subfigure[{~Uniform graph $G^u$. Edges are labeled with processing times,
   and heights in brackets}]{
\begin{tikzpicture}[>=latex']

\node[draw, circle] at (0,0) (2) {2};
\node[draw, circle] at (1,1) (1) {1};
\node[draw, circle] at (2,0) (3) {3};

\node[draw, circle] at (3.5,1) (4) {4};
\node[draw, circle] at (5,0)   (5) {5};
\node[draw, circle] at (2.5,-1.5) (7) {7};
\node[draw, circle] at (4.5,-1.5) (6) {6};

\draw[->, >=latex] (2) to node[above=5pt,left=-1pt]{$1\,(0)$} (1);
\draw[->, >=latex] (1) to node[above=5pt,right=-2pt]{$2\,(1)$} (3);
\draw[->, >=latex] (3) to node[above]{$3\,(0)$} (2);
\draw[->, >=latex] (3) to node[above=2pt,left=-1pt]{$2\,(1)$} (7);
\draw[->, >=latex] (7) to node[above]{$3\,(0)$} (6);
\draw[->, >=latex] (6) to node[below,right=2pt]{$1\,(1)$} (5);
\draw[->, >=latex] (5) to node[above=3pt,right=-2pt]{$5\,(0)$} (4);
\draw[->, >=latex] (4) to node[above]{$2\,(0)$} (3);

\node at (6,0) {};

\end{tikzpicture}
\label{fig:dep_graph}
}
\subfigure[~Graph $G(A)$]{
\begin{tikzpicture}[>=latex']

\node[draw, circle] at (0,0) (2) {2};
\node[draw, circle] at (1,1) (1) {1};
\node[draw, circle] at (2,0) (3) {3};

\node[draw, circle] at (3.5,1) (4) {4};
\node[draw, circle] at (5,0)   (5) {5};
\node[draw, circle] at (2.5,-1.5) (7) {7};
\node[draw, circle] at (4.5,-1.5) (6) {6};

\draw[->, >=latex] (1) to node[above]{$2$} (3);
\draw[->, >=latex] (2) to node[above]{$3$} (3);
\path (3) edge[loop above] node {$6$} (3);

\draw[->, >=latex] (3) to node[above]{$2$} (7);

\draw[->, >=latex,in=185,out=45] (7) to node[below]{$4$} (5);
\draw[->, >=latex,in=20,out=220] (5) to node[below]{$9$} (7);

\draw[->, >=latex] (5) to node[above=6pt,right=10pt]{$13$} (3);
\draw[->, >=latex] (6) to node[below,right=2pt]{$1$} (5);
\draw[->, >=latex] (4) to node[above]{$8$} (3);
\draw[->, >=latex] (4) to node[above=-5pt,left=-1pt]{$4$} (7);

\node at (6,0) {};

\end{tikzpicture}
\label{fig:associated_dep_graph}
}
\caption{Example of a set of tasks with restrictions}
\end{figure}

Call~$G^u$ {\em well-formed\/} if
it is strongly connected and does not contain a nonempty
     closed walk of height~$0$.
Call a schedule~$t$ an {\em earliest schedule\/} 
if it
     satisfies all restrictions specified by~$G^u$ and
     it is minimal with respect to the point-wise partial order on
     schedules.
Denote the maximum height in~$G^u$ by $\hat{h}$.
Cohen et al.~\cite{CMQV89} showed that the earliest schedule~$t$ for
well-formed~$G^u$ is
unique and fulfills%
\begin{equation}
t(i,n) = (A^{\otimes n} \otimes v)_i\label{eq:cohen}
\end{equation}%
for all $i\in \mathcal{T}$ and $n\ge 0$,  where~$v$ is a suitably
     chosen $(\hat{h} \cdot \lvert\mathcal{T}\rvert)$-dimensional
     max-plus vector and~$A$ a suitably chosen $(\hat{h} \cdot
     \lvert\mathcal{T}\rvert) \times (\hat{h} \cdot
     \lvert\mathcal{T}\rvert)$ max-plus matrix.
In case heights in~$G^u$ are binary, i.e., either~$0$ or~$1$, as in
     our example in Figure~\ref{fig:associated_dep_graph}, $A$ and $v$ are obtained as follows: For all $i,j
     \in \mathcal{T}$, $A_{i,j}$ is the maximum weight of nonempty
     walks~$W$ from~$i$ to~$j$ in~$G^u$, where all of~$W$'s edges have
     height~$0$, except for the last edge, which has height~$1$.
In case no such walk exists, $A_{i,j} = -\infty$.
For all $i \in \mathcal{T}$, $v_i$ is the maximum weight of walks~$W$
     from~$i$ in~$G^u$, where all of~$W$'s edges have height~$0$.
As an example the graph~$G(A)$ for the uniform graph in
     Figure~\ref{fig:dep_graph} is depicted in
     Figure~\ref{fig:associated_dep_graph}.
For this example we obtain the initial vector $v = (0,1,4,6,11,0,3)$.
We can, however, not directly apply our transience bounds on the
     graph~$G(A)$ obtained from $G^u$, since~$G(A)$ is not necessarily
     strongly connected, as it is the case for the example
in Figure~\ref{fig:associated_dep_graph}.

However, we present a transformation of~$G^u$ that yields a strongly
     connected graph~$G(A)$ in case of binary heights, and has the same
     earliest schedule as the original graph~$G^u$: For every
     restriction between tasks~$i$ and~$j$ in~$G^u$ one can add the
     {\em redundant restriction\/} $t(i,n)\geqslant t(j,n-1)+P_j$
     without changing the earliest schedule, since $t(j,n)\geqslant
     t(j,n-1)$ for all tasks~$j$ and $n\ge 1$.
With this transformation we obtain:

\begin{proposition}\label{prop:redundant}
If~$G^u$ is well-formed, has binary heights, and contains all redundant
restrictions, then~$A$ is irreducible.
\end{proposition}
\begin{proof}
It suffices to show that whenever there is an edge from~$i$ to~$j$ in~$G^u$, then it
also exists in~$G(A)$.
Because~$G^u$ contains all redundant restrictions, if there exists an edge
from~$i$ to~$j$, then there also exists an edge of height~$1$ from~$i$
to~$j$.
Hence there exists a walk of length~$1$ from~$i$ to~$j$ in~$G^u$ whose last
(and only) edge has height~$1$.
Hence, by definition of~$A$, the entry~$A_{i,j}$ is finite.
This concludes the proof.
\end{proof}

Figures~\ref{fig:dep_graph2} and~\ref{fig:associated_dep_graph2}
     depict the transformed graph~$G^u$ of the above example with
     redundant restrictions and its corresponding weighted
     graph~$G(A)$.
Observe that, in contrast to Figure~\ref{fig:associated_dep_graph},
     $G(A)$ is strongly connected in
     Figure~\ref{fig:associated_dep_graph2}.

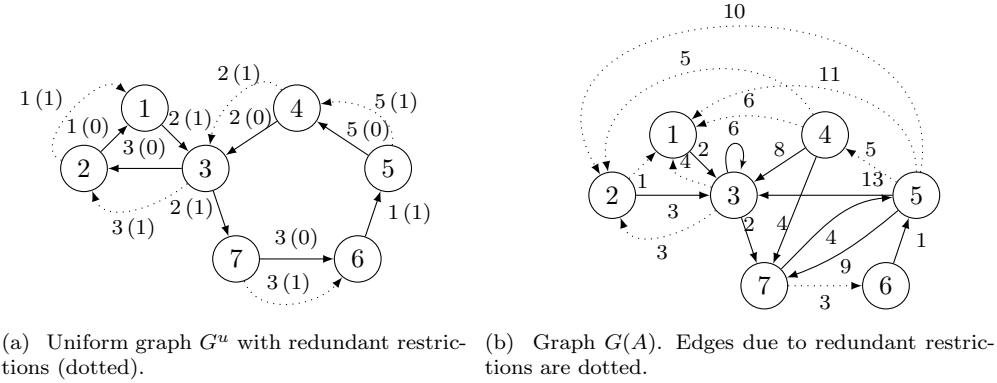
\begin{figure}[hbt]
\centering
\subfigure[{~Uniform graph $G^u$ with redundant restrictions (dotted).}]{
\begin{tikzpicture}[>=latex',scale=.8]

\node[draw, circle] at (0,0) (2) {2};
\node[draw, circle] at (1,1) (1) {1};
\node[draw, circle] at (2,0) (3) {3};

\node[draw, circle] at (3.5,1) (4) {4};
\node[draw, circle] at (5,0)   (5) {5};
\node[draw, circle] at (2.5,-1.5) (7) {7};
\node[draw, circle] at (4.5,-1.5) (6) {6};

\draw[->, >=latex] (2) to node[above=4pt,left=-2pt]{\footnotesize $1\,(0)$} (1);
\draw[dotted,->, >=latex,out=160,in=140] (2) to node[above=5pt,left=-1pt]{\footnotesize $1\,(1)$} (1);

\draw[->, >=latex] (1) to node[above=7pt,right=-6pt]{\footnotesize
  $2\,(1)$} (3);

\draw[->, >=latex] (3) to node[above]{\footnotesize $3\,(0)$} (2);
\draw[dotted,->, >=latex,out=220,in=290] (3) to node[below]{\footnotesize $3\,(1)$} (2);

\draw[->, >=latex] (3) to node[above=2pt,left=-1pt]{\footnotesize
  $2\,(1)$} (7);

\draw[->, >=latex] (7) to node[above]{\footnotesize $3\,(0)$} (6);
\draw[dotted,->, >=latex,out=290,in=230] (7) to node[above]{\footnotesize $3\,(1)$} (6);

\draw[->, >=latex] (6) to node[below,right=2pt]{\footnotesize
  $1\,(1)$} (5);

\draw[->, >=latex] (5) to node[above=3pt,right=-2pt]{\footnotesize
  $5\,(0)$} (4);
\draw[dotted,->, >=latex,out=80,in=10] (5) to node[above=3pt,right=-2pt]{\footnotesize
  $5\,(1)$} (4);

\draw[->, >=latex] (4) to node[above]{\footnotesize $2\,(0)$} (3);
\draw[dotted,->, >=latex,out=130,in=80] (4) to node[above]{\footnotesize $2\,(1)$} (3);

\node at (6,0) {};

\end{tikzpicture}
\label{fig:dep_graph2}
}~~%
\subfigure[~Graph $G(A)$. Edges due to redundant restrictions are dotted.]{
\begin{tikzpicture}[>=latex',scale=.8]

\node[draw, circle] at (0,0) (2) {2};
\node[draw, circle] at (1,1) (1) {1};
\node[draw, circle] at (2,0) (3) {3};

\node[draw, circle] at (3.5,1) (4) {4};
\node[draw, circle] at (5,0)   (5) {5};
\node[draw, circle] at (2.5,-1.5) (7) {7};
\node[draw, circle] at (4.5,-1.5) (6) {6};

\draw[dotted,->, >=latex] (2) to node[below]{\footnotesize $1$} (1);

\draw[->, >=latex] (1) to node[above]{\footnotesize $2$} (3);

\draw[->, >=latex] (2) to node[below]{\footnotesize $3$} (3);
\draw[dotted,->, >=latex,out=220,in=290] (3) to
node[below]{\footnotesize $3$} (2);

\path (3) edge[loop above] node {\footnotesize $6$} (3);

\draw[->, >=latex] (3) to node[above]{\footnotesize $2$} (7);

\draw[->, >=latex,in=185,out=45] (7) to node[below]{\footnotesize $4$}
(5);

\draw[dotted,->, >=latex] (7) to node[below]{\footnotesize $3$} (6);

\draw[->, >=latex,in=20,out=220] (5) to node[below]{\footnotesize $9$}
(7);

\draw[dotted,->, >=latex] (5) to node[above]{\footnotesize $5$} (4);

\draw[dotted,->, >=latex,out=160,in=270] (3) to
node[above]{\footnotesize $4$} (1);

\draw[dotted,->, >=latex,out=160,in=20] (4) to
node[above]{\footnotesize $6$} (1);

\draw[dotted,->, >=latex,out=90,in=40] (5) to
node[above]{\footnotesize $11$} (1);

\draw[dotted,->, >=latex,out=80,in=120,looseness=1.7] (5) to
node[above]{\footnotesize $10$} (2);

\draw[dotted,->, >=latex,out=120,in=100,looseness=1.2] (4) to
node[above]{\footnotesize $5$} (2);

\draw[->, >=latex] (5) to node[above=6pt,right=10pt]{\footnotesize $13$} (3);

\draw[->, >=latex] (6) to node[below,right=2pt]{\footnotesize $1$} (5);

\draw[->, >=latex] (4) to node[above]{\footnotesize $8$} (3);

\draw[->, >=latex] (4) to node[above=-5pt,left=-1pt]{\footnotesize $4$} (7);

\node at (6,0) {};

\end{tikzpicture}
\label{fig:associated_dep_graph2}
}
\caption{Transformation of $G^u$ in case of binary heights.}
\end{figure}

Because of \eqref{eq:cohen} and Proposition~\ref{prop:redundant} we
     may now directly apply Theorems~\ref{thm:nonexplorative}
     and~\ref{thm:explorative:nonprimitive} to (the strongly
     connected) graph~$G(A)$, obtaining upper bounds on the transients
     of the earliest schedule for~$G^u$.

For the given example, $\lVert v \rVert = 11$, the critical circuit is
     from node~$7$ to~$5$ and back, $\lambda=6.5$, $\lambda_{\nc}=6$,
     $\Delta_{\nc} = 8$, $\delta=1$, $\hat{g}=2$, $\hat{\gamma}=2$,
     $\hat{ep} = 0$, and we obtain a critical bound of~$106$.
Since the critical bound dominates both the repetitive and explorative
     bound of Theorems~\ref{thm:nonexplorative}
     and~\ref{thm:explorative:nonprimitive} respectively, $106$ is an
     upper bound on the transient of the earliest schedule.
The discrepancy to the transient of the earliest schedule, which
     is~$1$, stems from the fact that the critical bound is overly
     conservative for this example.

Bounds in terms of the parameters of the original uniform graph~$G^u$ can be
     derived as well by relating graph parameters of~$G^u$ to
     parameters of $G=G(A)$.
For that purpose, we denote by $\delta(G^u)$ and $\Delta(G^u)$ the
     minimum and maximum weight of an edge in~$G^u$, respectively.
From the definition of max-plus matrix~$A$ and initial vector~$v$, it
     immediately follows that in case of binary heights, $N= \lvert
     \mathcal{T}\rvert$, $\lVert v \rVert \le (\lvert
     \mathcal{T}\rvert-1)\cdot\Delta(G^u)$, $\Delta(G) \le
     \lvert\mathcal{T}\rvert\cdot\Delta(G^u)$, $\delta(G) \ge
     \delta(G^u)$, 
\begin{equation*}
\lambda(G) = \max \{ p(C)/h(C) \mid C \text{ is a closed walk in }
     G^u\}\enspace,%
\end{equation*}%
$\lambda_{\nc}(G)$ is at most the second largest $p(C)/h(C)$ of closed
     walks~$C$ in $G^u$, and $\hat{g}(G)$ is at most the number of
     links with height~$1$ in closed walks~$C$ in $G^u$ with maximum
     $p(C)/h(C)$.
As a consequence of the above bounds and the bound stated
     in~\eqref{eq:int_bound} for integer matrices, the transient is
     in~$O((\Delta(G)-\delta(G))\cdot \lvert \mathcal{T}\rvert^3) =
     O(\lvert \mathcal{T}\rvert^4)$, assuming constantly bounded~$\delta(G^u)$
     and~$\Delta(G^u)$.
To the best of our knowledge, this is the first asymptotic bound on
     the transient of an earliest schedule with tasks~$\mathcal{T}$
     and binary heights.
It is an open problem whether this bound in~$\lvert \mathcal{T}\rvert$ is asymptotically tight.

\subsection{Synchronizers}\label{subsec:sync}

Even and Rajsbaum~\cite{even:rajs} presented a transience bound for a
     network synchronizer in a system with constant integer communication
     delays.
They considered a variant of the $\alpha$-synchronizer~\cite{Aw85} in
     a centrally clocked distributed system of $N$~processes that
     communicate by message passing over a strongly connected network
     graph~$G$.
Each link has constant transmission delay, specified in terms of
     central clock ticks.
Processes execute the $\alpha$-synchronizer after an initial boot-up phase: After
     receiving round~$n$ messages from all neighbors, a process
     proceeds to round $n+1$ and broadcasts its round $n+1$ message.
Denote by $t(n)$ the vector such that $t_i(n)$ is the clock tick at
     which process~$i$ broadcasts its round~$n$ message.
Even and Rajsbaum showed that the synchronizer becomes periodic by time
	$B_\mathrm{ER} = l_0+ 2N^2 + N$, where~$l_0$ is an
     upper bound on the length of maximum weight walks with only
     non-critical nodes.
It is easily checked that~$l_0$ is always greater or equal to our critical
bound~$\Bcnc$.

One can show that~$t(n)$ is in fact a max-plus linear system.
More precisely, $t(n) = A^{\otimes n} \otimes t(0)$, where~$A$ is the
     adjacency matrix of the network graph~$G$.
Our bounds hence directly apply, and we obtain a repetitive bound on
     the transient of~$(t(n))_{n\ge 0}$ that is strictly less than
     $\max\{l_0, 2N^2-N\}$, and thus strictly less than Even and
     Rajsbaum's bound~$B_\mathrm{ER}$.

As an example, let us consider the ``$\ell$-sized cherry'' graph
     family~$H_{\ell,c}$, with $\ell\ge 2$ and $c\ge 1$, introduced by
     Even and Rajsbaum~\cite{even:rajs}.
Each weighted graph~$H_{\ell,c}$ contains $N=4\ell$ nodes and is
     constructed as follows: Let~$\hat{C}$ and~$C$ be two cycles of
     length~$\ell$ and~$\ell+1$ respectively, with edge weights~$3c$,
     except for one link per cycle with weight~$3c+1$.
There exists for both~$\hat{C}$ and~$C$ a path of length~$\ell$ to a
     distinct node~$s$, and an antiparallel path back.
Hereby the edges in the path from~$s$ to~$C$ and from~$s$ to~$\hat{C}$
     have weight~$c$, the edges in the path from~$\hat{C}$ to~$s$ have
     weight~$3c$, and from~$C$ to~$s$, $4c$.

\begin{figure}[hbt]
\centering
\begin{tikzpicture}[>=latex',scale=1.2]
        \node (x) at (-1.5,0.5) {$\hat{C}$};
        \node (x) at (9.1,0.5) {$C$};
	\node[shape=circle,draw] (u) at (0,0) {};
        \node[shape=circle,draw] (u1) at (-1,0.5) {};
        \node[shape=circle,draw] (u2) at (-1,-0.5) {};
        
        \node[shape=circle,draw] (v) at  (7.2,0) {};
        \node[shape=circle,draw] (v1) at (8,-0.5) {};
        \node[shape=circle,draw] (v2) at (8,0.5) {};
        \node[shape=circle,draw] (vb) at (8.7,0) {};

        \node[shape=circle,draw] (1) at (1.2,0.2) {};
        \node[shape=circle,draw] (2) at (2.4,0.4) {};
        \node[shape=circle,draw] (3) at (3.6,0.6) {$\scriptstyle s$};

        \node[shape=circle,draw] (4) at (4.8,0.4) {};
        \node[shape=circle,draw] (5) at (6,0.2) {};

        \draw[->] (u1) to node[below=0.1mm] {$6$} (u);
        \draw[->] (u2) to node[left=0.1mm] {$6$} (u1);
        \draw[->] (u) to node[below=0.1mm] {$7$} (u2);

        \draw[->] (v1) to node[below=0.1mm] {$6$} (v);
        \draw[->] (vb) to node[below=0.1mm] {$6$} (v1);
        \draw[->] (v2) to node[below=0.1mm] {$6$} (vb);
        \draw[->] (v) to node[below=0.1mm] {$7$} (v2);

        \draw[->, bend left] (u) to node[above=1mm] {$6$} (1);
        \draw[->, bend left] (1) to node[above=1mm] {$6$} (2);
        \draw[->, bend left] (2) to node[above=1mm] {$6$} (3);
        \draw[->, bend left] (3) to node[above=1mm] {$2$} (4);
        \draw[->, bend left] (4) to node[above=1mm] {$2$} (5);
        \draw[->, bend left] (5) to node[above=1mm] {$2$} (v);

        \draw[<-, bend right] (u) to node[below=0.5mm] {$2$} (1);
        \draw[<-, bend right] (1) to node[below=0.5mm] {$2$} (2);
        \draw[<-, bend right] (2) to node[below=0.5mm] {$2$} (3);
        \draw[<-, bend right] (3) to node[below=0.5mm] {$8$} (4);
        \draw[<-, bend right] (4) to node[below=0.5mm] {$8$} (5);
        \draw[<-, bend right] (5) to node[below=0.5mm] {$8$} (v);
\end{tikzpicture}
\caption{Graph $H_{3,2}$ }
\label{fig:H32}
\end{figure}
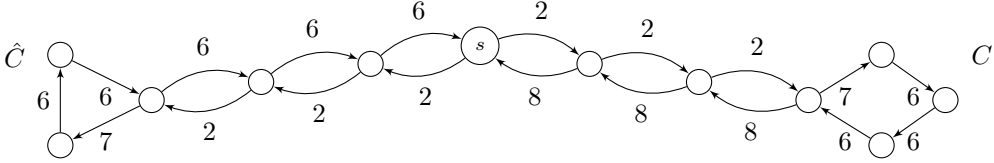

Observing that the nodes of $\hat{C}$ are the critical nodes,
     $\Delta=4c$, $\delta=c$, $N=4\ell$, $\lambda=3c+1/\ell$, and
     $l_0=112c\ell^3-16\ell^3-12c\ell^2+4\ell-1$, Even and Rajsbaum's
     bound is  
\begin{equation}
(112c-16)\ell^3+(32-12c)\ell^2+8\ell-1\notag\enspace,
\end{equation}
resulting in an upper bound of~$5711$ on the transient in case of~$H_{3,2}$.
Since $\Delta_\nc=\Delta$ and $\lambda_\nc=3c+1/(\ell+1)$, we obtain for
     the critical bound $B_\cc = 3c\ell(\ell+1)(N-1) = 12c\ell^3+9c\ell^2-3c\ell$.
Moreover for the critical subgraph~$G_\cc$, the maximum girth of strongly
     connected components of $G_\cc$ is $\hat{g}=\ell$.
Thereby we may bound the transient of $(t(n))_{n\ge 0}$ with
     Theorem~\ref{thm:nonexplorative} by  
\begin{equation}
\max\{B_\cc, 2\ell N-\ell-1 \} = \max\{B_\cc, 8\ell^2-\ell-1 \} =  12c\ell^3+9c\ell^2-3c\ell\notag\enspace,
\end{equation}%
resulting in an upper bound of~$792$ on the transient in case of~$H_{3,2}$.

Since Even and Rajsbaum express transmission delays with respect to 
     a discrete global clock, all weights are integers.
Both our transience bounds are in $O(\lVert A\rVert \cdot N^3)$.
The example graph family shows that this is asymptotically tight since Even and
	Rajsbaum proved that the transient for graph $H_{c,\ell}$ is in
	$\Omega(c \cdot \ell^3)=\Omega(\lVert A\rVert \cdot N^3)$.
An adapted example graph family shows asymptotic tightness of our bounds in the general case.

\subsection{Full Reversal routing and scheduling}\label{subsec:fr}

Link reversal is a versatile algorithm design paradigm, which was, in
     particular,
     successfully applied to routing~\cite{GB87} and
     scheduling~\cite{BG89}.
Charron-Bost et al.~\cite{pr:sirocco} showed that the analysis of a
     general class of link reversal algorithms can be reduced to the
     analysis of Full Reversal, a particularly simple algorithm on
     directed graphs.

The Full Reversal algorithm comprises only a single rule: Each sink reverses all its
     (incoming) edges.
Given a weakly connected initial graph~$G_0$ without antiparallel edges, we consider a {\em
     greedy\/} execution of Full Reversal as a sequence $(G_t)_{t\ge
     0}$ of graphs, where~$G_{t+1}$ is obtained from $G_t$ by
     reversing the edges of {\em all\/} sinks in~$G_t$.
As no two sinks in~$G_t$ can be adjacent, $G_{t+1}$ is well-defined.
For each $t\ge 0$ we define the {\em work vector}~$W(t)$ by setting
     $W_i(t)$ to the  number of reversals of node~$i$ until
     iteration~$t$, i.e., the number of times node~$i$ is a sink in
     the execution prefix $G_0,\dots,G_{t-1}$.

Charron-Bost et al.\ \cite{fr:sirocco} have shown that the sequence of
     work vectors can be described as a {\em min-plus\/} linear
     dynamical system.
Min-plus algebra is a variant of max-plus algebra, using $\min$
     instead of $\max$.
Denoting by $\otimes'$ the matrix multiplication in min-plus algebra,
 Charron-Bost et
     al.\ established that $W(0) = 0$ and
     $W(t+1) = A \otimes' W(t)$, where $A_{i,j} = 1$ and $A_{j,i} = 0$ if
     $(i,j)$ is an edge of the initial graph $G_0$; otherwise $A_{i,j} = +\infty$.
Observe that the latter min-plus recurrence is equivalent to $-W(t+1) = (-A) \otimes
     (-W(t))$ where $-A$ is an integer max-plus matrix with
     $\Delta_\nc\in\{ 0,-1\}$ and $\delta = -1$.

\subsubsection{Full Reversal routing}
In the routing case, the initial graph~$G_0$ contains a nonempty set of {\em
     destination nodes}, which are characterized by having a self-loop.
The initial graph without these self-loops is required to be weakly
     connected and acyclic~\cite{fr:sirocco,GB87}.
It was shown that for such initial graphs, the execution terminates
     (eventually all~$G_t$ are equal), and after termination, the
     graph is destination-oriented, i.e., every node has a walk to
     some destination node.
We now show how the previously known results that the termination time of Full
Reversal routing is quadratic in general~\cite{BT05} and linear in
trees~\cite{fr:sirocco} directly follows from both
Theorem~\ref{thm:nonexplorative} and Theorem~\ref{thm:explorative:nonprimitive}.

The set of critical nodes is equal to the set of destination nodes and
     each strongly connected component of~$G_\cc$ consists of a single
     node.
Hence $\lambda=0$ and $\lambda_\nc \le -{1}/{N_\nc} \le -{1}/(N-1)$,
     i.e.,~$(N-1)^2$ is an upper bound on the critical bound.
Since $\hat{g} = 1$, we obtain from Theorem~\ref{thm:nonexplorative},
     for $N\geqslant3$, that the termination time is at most
     $(N-1)^2$, which improves on the asymptotic quadratic bound given by Busch
and Tirthapura~\cite{BT05}.

If the undirected support of initial graph~$G_0$ without the self-loop
     at the destination nodes is a {\em tree}, we can use our bounds
     to give a new proof that the termination time of Full Reversal
     routing is linear in~$N$ \cite[Corollary~5]{fr:sirocco}.
In that particular case either $\lambda_\nc=-{1}/{2}$ or
     $\lambda_\nc=-\infty$.
In both cases the critical bound is at most $2(N-1)$.
Both Theorem~\ref{thm:nonexplorative} and
Theorem~\ref{thm:explorative:nonprimitive} yield the linear bound
     $2(N-1)$, whereas Hartmann and Arguelles arrive at $2N^2$.

\subsubsection{Full Reversal scheduling}

When using the Full Reversal algorithm for scheduling, the undirected
     support of the weakly connected initial graph~$G_0$ is
     interpreted as a conflict graph: nodes model processes and an
     edge between two processes signifies the existence of a shared
     resource whose access is mutually exclusive.
The direction of an edge signifies which process is allowed to use the
     resource next.
A process waits until it is allowed to use all its resources---that
     is, it waits until it is a sink---and then performs a step, that
     is, reverses all edges to release its resources.
To guarantee liveness, the initial graph~$G_0$ is required to be
     acyclic.

Contrary to the routing case, strongly connected components of the critical
subgraph 
     have at least two nodes, because there are no self-loops.
By using~\eqref{eq:int_bound}, we get ${N^2(N-1)}/{4}$ as an
     upper bound on our critical bound, which
     shows that the transient for Full
     Reversal scheduling is at most cubic in the number~$N$ of
     processes.
Malka and Rajsbaum \cite[Theorem~6.4]{malka:rajs} proved by reduction
     to Timed Marked Graphs that the transient is at most in the order
     of~$O(N^4)$.
Thus, our bounds allow to improve this asymptotic result by an
     order of $N$.

In the case of Full Reversal scheduling on {\em trees\/} we even
     obtain a bound linear in $N$: In this case it holds that
     $\lambda=-{1}/{2}$, and $\lambda_\nc=-\infty$.
Thus the critical bound is~$N$.
Further, $G_\cc = G$ and $\hat{g}=2$.
Both Theorem~\ref{thm:nonexplorative} and
Theorem~\ref{thm:explorative:nonprimitive} thus imply that $4N-3$ is an upper
     bound on the transient of Full Reversal scheduling on trees,
     which is linear in~$N$.
This was previously unknown.
By contrast Hartmann and Arguelles again obtain the quadratic bound
     of~$2N^2$.

\section*{Acknowledgments}

The authors would like to thank %
Fran\c{c}ois Baccelli, Anne Bouillard, Philippe Chr\'{e}tienne, and
     Serge{\u \i} Sergeev for helpful discussions.

\end{document}